\documentclass[12pt]{amsart}
\usepackage[margin=1in]{geometry}
\usepackage{amsmath}
\usepackage[latin1]{inputenc}
\usepackage{fancyhdr}
\usepackage{hyperref}
\usepackage{stmaryrd}
\usepackage{epsfig}
\usepackage{tikz}
\usepackage{euscript}
\usepackage{color}
\usepackage{epstopdf}
\usepackage{pgfplots}
\usetikzlibrary{calc}
\usepackage[scaled=1.5]{helvet}		
\usepackage{courier}			
\usepackage{bbm}
\usepackage[latin1]{inputenc}
\usepackage{listings}
\usepackage{graphicx}
\usepackage{amsmath}
\usepackage{amssymb}
\usepackage{bbm}
\newtheorem{theorem}{Theorem}

\newtheorem{definition}[theorem]{Definition}

\newtheorem{lemma}[theorem]{Lemma}

\newtheorem{proposition}[theorem]{Proposition}
\newtheorem{remark}[theorem]{Remark}

\usepackage{listings}
\usepackage [all,arc,curve,color, arrow, matrix,frame]{xy}
\usepackage[english]{babel}
\usepackage{pgf,tikz}
\usepackage{mathrsfs}
\usetikzlibrary{arrows}

\newcommand{\Z}{{\mathbb Z}}
\newcommand{\R}{{\mathbb R}}
\newcommand{\C}{{\mathbb C}}
\newcommand{\N}{{\mathbb N}}

\newcommand{\spann}{{\mbox{span}}}

\newcommand{\idty}{{\mathbbm{1}}}

\numberwithin{equation}{section}
\numberwithin{theorem}{section} 
\numberwithin{footnote}{section}
\pgfplotsset{compat=1.10}

\begin{document}

\definecolor{ududff}{rgb}{0.30196078431372547,0.30196078431372547,1.}

\title[XXZ systems on graphs]{Droplet states in quantum XXZ spin systems\\ on general graphs}

\author[C. Fischbacher]{Christoph Fischbacher$^1$}
\address{$^1$ Department of Mathematics\\
University of Alabama at Birmingham\\
Birmingham, AL 35294, USA}
\email{cfischb@uab.edu}

\author[G. Stolz]{G\"unter Stolz$^2$}
\address{$^2$ Department of Mathematics\\
University of Alabama at Birmingham\\
Birmingham, AL 35294, USA}
\email{stolz@uab.edu}

\date{}

%
\begin{abstract}

We study XXZ spin systems on general graphs. In particular, we describe the formation of droplet states near the bottom of the spectrum in the Ising phase of the model, where the Z-term dominates the XX-term. As key tools we use  particle number conservation of XXZ systems and symmetric products of graphs with their associated adjacency matrices and Laplacians. Of particular interest to us are strips and multi-dimensional Euclidean lattices, for which we discuss the existence of spectral gaps above the droplet regime. We also prove a Combes-Thomas bound which shows that the eigenstates in the droplet regime are exponentially small perturbations of strict (classical) droplets.

\end{abstract}

\maketitle

\tableofcontents

\section{Introduction}

The infinite XXZ quantum spin chain is, at least in principle, fully diagonalizable by the Bethe ansatz, e.g.\ \cite{Bethe,Thomas1977,BabbittThomas1977,BabbittThomas1978,BabbittThomas1979,BabbittGutkin1990,Gutkin1993,Gutkin2000,Borodinetal}. However, having explicit expressions available for all its generalized eigenfunctions does not readily expose the quantitative information needed to study further questions related to the XXZ chain such as finite volume effects or perturbative situations (e.g., the sensitivity of the model to exterior fields). It also does not reflect the behavior of important physical quantities such as quantum dynamics or entanglement properties of eigenstates in a straightforward way.

More detailed information about the bottom of the spectrum of the XXZ chain has been found in its Ising phase, where the Z-term dominates the XX-term of the Hamiltonian. Not only does this phase lead to a non-degenerate gapped ground state \cite{KN1,KN2}, but the first band above the ground state gap is found to be generated by states consisting of small perturbations of {\it droplets}, i.e., a single cluster of spins of a fixed orientation \cite{Starr,NS,FS}. In finite volume these droplets appear in connection with kink boundary conditions \cite{NSS}.

This droplet regime is quite stable under perturbations, which has recently allowed to prove that exposure of the Ising phase XXZ chain to a random field in the Z direction leads to a broad range of localization properties within the droplet spectrum \cite{BW1,EKS1,EKS2,BW2}, thus providing the first mathematically rigorous results towards the existence of a physically expected many-body localized phase in the disordered XXZ chain.

Our main goal here is to demonstrate that the occurrence of the droplet phase is not limited to the special case of the XXZ chain, but appears much more generally for {\it quasi-one-dimensional} XXZ systems in the Ising phase. As our set-up we will consider the XXZ model on a general undirected countable graph ${\mathcal G}$. This includes, in particular, the quasi-one-dimensional case of strips of arbitrary width, which will be contrasted with XXZ systems on multi-dimensional lattices (where we will focus on dimension two). In particular, the general model considered here is not exactly solvable, a property which is not needed for our methods.

The first crucial feature to be exploited is that XXZ models are particle number preserving, meaning that the subspaces generated by states with an arbitrary but fixed number $N$ of (say) down-spins are invariant under the XXZ Hamiltonian. The restrictions of the Hamiltonian to the $N$-particle subspaces turn out to be discrete Schr\"odinger-type operators on the {\it $N$-th symmetric product graph} ${\mathcal G}_N$ of ${\mathcal G}$. This concept and some of its graph theoretic implications have been used in the literature on spin systems before, at least for the isotropic XXX (or Heisenberg) model \cite{AGRR,O2,O1} (and \cite{BarghiPonomarenko2009,AIP} provide counterexamples to an iso-spectrality problem posed in \cite{AGRR}). 

We feel that this connection between quantum spin systems and graph theory still has much potential remaining to provide further insights, some of which we discuss here. Interestingly (but not too surprisingly), symmetric product graphs have been introduced and studied in the graph theory literature, seemingly independently, under several other names such as $N$-subgraph graphs, $N$-tuple vertex graphs and $N$-token graphs, e.g.\ \cite{Johns, Alavi1, Zhuetal, Alavi2, Fabila, Carballosa, Rivera}.

Physically, the reduced $N$-particle operators describe hard core bosons with attractive next-neighbor interaction. In the Ising phase these interactions (resulting from the $Z$-term) dominate the hopping terms (resulting from the $XX$-term). This is the second crucial feature exploited in our work --- and in earlier works on the special case of the chain --- as it makes the particles ``sticky'' at low energies and thus leads to the formation of spin droplets.

In Section~\ref{sec:XXZ} we start by introducing some graph theoretic concepts, including symmetric products of graphs and a known formula for the geodesic distance on these products (for which we provide a proof in Appendix~\ref{AppA}). Then we introduce the XXZ system on a general connected undirected graph $\mathcal G$ and show that its restrictions to the $N$-particle subspaces are given by Schr\"odinger-type operators on the symmetric products ${\mathcal G}_N$ of $\mathcal G$, with potential given by the degree function on the graph ${\mathcal G}_N$. The latter is given for any vertex of ${\mathcal G}_N$ by the {\it edge surface measure} of the corresponding $N$-particle configuration in $\mathcal G$. This leads to a close connection between the Ising phase of XXZ systems and the edge-isoperimetric problem on the graph (Section~\ref{sec:droplets}). In this view droplets are given the natural geometric interpretation of solutions of the isoperimetric problem, i.e., configurations of minimal surface measure at given particle number (or volume). This leads us to the definition of {\it droplet spectrum} for a general XXZ system in the Ising phase.

The remainder of the paper is devoted to two questions: (i) Is the droplet spectrum (in sufficiently strong Ising phase) separated by a gap from higher spectral bands of the XXZ system? (ii) In what sense are the eigenstates to energies in the droplet spectrum approximated by droplet states, i.e., by spin product states with a single connected cluster of (in our choice) down-spins?

In Section~\ref{sec:abstractgap} we develop a general approach to the proof of the existence of spectral gaps for Schr\"odinger type operators on graphs, tailored to our application. It will allow to treat the hopping terms in the $N$-particle operators as perturbations of the potential term. We will exploit that, while the graphs ${\mathcal G}_N$ have ``bulk degree'' growing linearly in $N$, the local degree of vertices in ${\mathcal G}_N$ corresponding to droplet configurations is smaller than the bulk degree, growing as the surface area rather than the volume of droplets.

This means, essentially, that the droplet regime at the bottom of the spectrum can be expected to be uniform in the particle number $N$ for graphs $\mathcal{G}$ which are ``one-dimensional'' in the sense that droplets have surface area uniformly bounded in their volume. For higher-dimensional graphs the surface area grows with the volume of a droplet, leading to the conclusion that the droplet regime at the bottom of the spectrum only reflects small particle numbers.

In Section~\ref{sec:examples} we discuss this in some detail for two prototypical examples, strips of arbitrary (but fixed) width $M$, as well as the two-dimensional Euclidean lattice $\Z^2$. Here we also include some new results on the existence of higher order spectral gaps for the case of a chain, i.e., the strip with $M=1$. For completeness of presentation, we include some results on the solution of the edge-isoperimetric problem for these examples in Appendix \ref{sec:AppB}.

We limit our discussion of special cases to these two examples, as this should suffice to point out the principal difference between one-dimensional and multi-dimensional systems. In principle, other examples of quasi-one-dimensional graphs, such as ``tubes'' or graphs found from the chain or strip by adding local graph decorations, or graphs given by higher-dimensional Euclidean lattices, could be analyzed and should show similar features. However, the details would get more tedious due to the subtleties of solving the corresponding edge-isoperimetric problems.

In many ways our main results are Theorem~\ref{thm:CTgen} in Section~\ref{sec:CTbound} and the subsequent discussion of decay of eigenstates and spectral projections for energies in the droplet spectrum in Section~\ref{thm:CTgen}, as this shows how well these states are approximated by strict droplet configurations. In fact, the approximation turns out to be exponentially close in the graph distance on the symmetric product graph ${\mathcal G}_N$. Here the central Theorem~\ref{thm:CTgen} takes the form of a Combes-Thomas-type bound on the Green function with respect to the graph distance on ${\mathcal G}_N$, generalizing an argument for the chain in \cite{EKS1} (see also \cite{BW1} for a related result).

While we focus on  the droplet regime, we mention that all our results can be used to study the structure of the spectrum and eigenstates of XXZ systems at higher energies, where the eigenstates correspond to ``multi-droplets'', i.e., are close to spin product states with multiple connected components of down-spins.

Also, by not using any of the specific properties behind the exact solvability of the XXZ chain via the Bethe ansatz, our results allow the addition of positive exterior magnetic fields in the Z-direction to the model and work equally well in finite and infinite volume. This, in particular, will be exploited in planned future work such as the further study of localization properties of XXZ systems in random fields.

\section{The XXZ Hamiltonian on general graphs} \label{sec:XXZ}

\subsection{Basic graph theoretic concepts and symmetric products}

In the following, let $\mathcal{G}=(\mathcal{V,E})$ denote an undirected non-trivial graph with countable (finite or infinite) vertex set $\mathcal{V}$ and edge set $\mathcal{E}\subset\{\{x,y\}: x,y\in \mathcal{V}, x\neq y\}$. We assume that $\mathcal{G}$ is connected, so that the graph distance $d(\cdot,\cdot)$ on $\mathcal{G}$, i.e., the smallest number of edges connecting two vertices, is an integer-valued metric on $\mathcal{V}$. We also assume that $\mathcal{G}$ has bounded degree, i.e., that $d:=\sup_{x\in \mathcal V} d(x)<\infty$, where $d(x) := |\{y\in \mathcal{V}: \{x,y\} \in \mathcal{E}\}| = |\{y\in \mathcal{V}: d(x,y)=1\}|$ is the local degree of $x\in {\mathcal V}$. Here and below $|\cdot|$ denotes cardinality. 

\begin{definition} \normalfont For every finite $X\subset\mathcal{V}$, let us introduce the following graph theoretic concepts:
\begin{itemize}
\item The {\it total degree} $d_{tot}(X)$ given by
\begin{equation} 
d_{tot}(X)=\sum_{x\in X}d(x)\:.
\end{equation}
\item The {\it edge boundary} $\partial X$ of $X$ given by
\begin{equation} \label{eq:edgeboundary}
\partial X=\{\{x,y\}\in\mathcal{E}: x\in X, y\notin X\}\:.
\end{equation}
The {\it surface measure} $S(X)$ is then given by $S(X)=|\partial X|$. 
\item The {\it interaction potential} $W(X)$ given 
by the number of edges between elements of $X$,
\begin{equation}
W(X) = |\{\{x,y\} \subset X: \{x,y\} \in \mathcal{E}\}|.
\end{equation} 
Alternatively, explaining the use of the term interaction potential, one can write
\begin{equation} 
W(X)=\sum_{\{x,y\} \subset X}Q(x,y)\:,
\end{equation}
where $Q$ is the attractive next-neighbor pair interaction
\begin{equation}
Q(x,y)=\begin{cases} 1, \quad&\text{if}\quad \{x,y\}\in \mathcal{E},\\
0,\quad&\text{else.}
\end{cases}
\end{equation}
\end{itemize}
\end{definition}

Note that for finite $X\subset \mathcal{V}$ the quantities $d_{tot}(X)$, $S(X)$ and $W(X)$ are all finite. Moreover, they satisfy the following simple relation (which is well known, at least for regular graphs, e.g.\ \cite{Bezrukov}):
\begin{lemma} \label{lemma:graphhelp}
For $S(X)$, $d_{tot}(X)$ and $W(X)$ as defined above we have 
\begin{equation} \label{eq:georel}
S(X) = d_{tot}(X) - 2W(X)\:.
\end{equation}
\end{lemma}

\begin{proof}
We have 
\begin{equation} \label{eq:S(X)}
S(X)=|\partial X|=\sum_{x\in X}\left(\sum_{y\notin X:\{x,y\}\in\mathcal{E}}1 \right)\:.
\end{equation}
as well as
\begin{equation} \label{eq:totdegree}
d_{tot}(X)=\sum_{x\in X}\left(\sum_{y\in\mathcal{V}:\{x,y\}\in\mathcal{E}}1 \right)=\sum_{x\in X}\left(\sum_{y\in X:\{x,y\}\in\mathcal{E}}1 \right)+\sum_{x\in X}\left(\sum_{y\notin X:\{x,y\}\in\mathcal{E}}1 \right)\:.
\end{equation}
Finally, $W(X)$ can be expressed as
\begin{equation} \label{eq:potential}
W(X)=\sum_{\{x,y\}\subset X: \{x,y\}\in\mathcal{E}}1=\frac{1}{2}\sum_{x\in X}\left(\sum_{y\in X:\{x,y\}\in\mathcal{E}}1\right)\:,
\end{equation}
where the prefactor of $1/2$ on the right hand side of \eqref{eq:potential} accommodates for the fact that each edge is counted twice in the summation appearing there. Combined, this yields \eqref{eq:georel}.
\end{proof}

Given a graph $\mathcal{G}=(\mathcal{V,E})$, let us now define its \emph{$N$-th symmetric product.}

\begin{definition} Let $\mathcal{G}=(\mathcal{V,E})$ be an undirected connected and countable graph. For any $N\in\N$ such that $N\leq |\mathcal{V}|$, we then define the {\bf{$N$-th symmetric product}} $\mathcal{G}_N=(\mathcal{V}_N,\mathcal{E}_N)$ of $\mathcal{G}$ to be the graph with vertex set 
\begin{equation} 
\mathcal{V}_N=\{X\subset\mathcal{V}:|X|=N\}
\end{equation}
and edge set 
\begin{equation} 
\mathcal{E}_N=\{\{X,Y\}:X,Y\in\mathcal{V}_N, X\triangle Y\in \mathcal{E}\}\:.
\end{equation}
Here $X\triangle Y=(X\setminus Y)\cup (Y\setminus X)$ is the symmetric difference of the sets $X$ and $Y$. Moreover, let $d_N(X,Y)$ denote the graph distance between two vertices $X,Y\in\mathcal{V}_N$ on $\mathcal{G}_N$.
\end{definition}

\begin{remark} {\rm It will be consistent with our later considerations to define $\mathcal{G}_0$ to be the trivial graph with only one vertex, i.e.\ $\mathcal{G}_0:=(\{\emptyset\},\emptyset)$.}
\end{remark}

\begin{remark} \label{neighbor}
{\rm Note that $(X,Y) \in \mathcal{E}_N$ means that $Y$ is found from $X$ by moving one element of $X$ to an unoccupied next neighbor in $\mathcal{G}$, while all other elements stay in place. This can be done in exactly $S(X)$ ways. Thus, for a vertex $X\in {\mathcal V}_N$, $S(X)$ is its {\it local degree} in the graph ${\mathcal G}_N$. As ${\mathcal G}$ has degree bounded by $d$, \eqref{eq:georel} implies that each ${\mathcal G}_N$ has degree bounded by $Nd$.}
\end{remark}

\begin{remark}
{\rm Using labelings $X=\{x_1,\ldots,x_N\}$ and $Y=\{y_1,\ldots,y_N\}$, the graph distance $d_N(\cdot,\cdot)$ can be expressed in terms of $d(\cdot,\cdot)$ by
\begin{equation} \label{distN}
d_N(X,Y) = \min_{\pi \in D_N} \sum_{j=1}^N d(x_j, y_{\pi(j)}).
\end{equation}
In particular, all ${\mathcal G}_N$ are connected.
While \eqref{distN} is known, e.g.\ \cite{O1}, and a proof can be found by combining arguments in Chapter V of \cite{Johns}, we include a short self-contained proof in Appendix~\ref{AppA} for the reader's convenience.}
\end{remark}

\subsection{The XXZ Hamiltonian on finite graphs} \label{XXZfinite}

 Let us denote the canonical basis of $\C^2$ by
\begin{equation}
e_0 := \begin{pmatrix} 1\\0 \end{pmatrix}  \quad \mbox{and} \quad e_1 := \begin{pmatrix} 0\\1 \end{pmatrix}\:,
\end{equation}
interpreted as ``up-spin" and ``down-spin", respectively. 

We will initially consider {\it finite} graphs $\mathcal{G} =(\mathcal{V},\mathcal{E})$ and define the Hilbert space 
\begin{equation} 
\mathcal{H}_\mathcal{G}=\bigotimes_{x\in\mathcal{V}}\mathcal{H}_x=\bigotimes_{x\in\mathcal{V}}\C^2\:,
\end{equation}
which has dimension $2^{|\mathcal{V}|}$. 
For $x, y \in \mathcal{V}$ with $x\neq y$ we define the two-site Hamiltonian 
\begin{eqnarray} \label{eq:twosite}
h_{xy} & = & -\frac{1}{\Delta} (S_x^1 S_y^1 + S_x^2 S_y^2) - S_x^3 S_y^3 +\frac{1}{4} \idty \\ 
& = & -\frac{1}{\Delta}\left(\vec{S}_x\vec{S}_y-\frac{1}{4}\idty\right)-\left(1-\frac{1}{\Delta}\right)\left(S^3_xS^3_y-\frac{1}{4}\idty\right)\:, \notag
\end{eqnarray}
acting on the local Hilbert space $\mathcal{H}_x\otimes\mathcal{H}_y=\C^2\otimes\C^2$. Here $\vec{S}=(S^1,S^2,S^3)$ are the spin matrices
\begin{equation}
S^1 = \frac{1}{2} \begin{pmatrix} 0 & 1 \\ 1 & 0 \end{pmatrix}, \quad S^2 = \frac{1}{2} \begin{pmatrix} 0 & -i \\ i & 0 \end{pmatrix}, \quad S^3 = \frac{1}{2} \begin{pmatrix} 1 & 0 \\ 0 & -1 \end{pmatrix}
\end{equation} 
and $\vec{S}_x\vec{S}_y := S_x^1 S_y^1 + S_x^2 S_y^2 + S_x^3 S_y^3$.
At this point it is also convenient to introduce the spin lowering operator $S^-$ and the spin raising operator $S^+$:
\begin{align} 
 S^- &:= S^1 -iS^2 = \begin{pmatrix} 0 & 0 \\ 1 & 0 \end{pmatrix}\:,\\
 S^+ &:= S^1+iS^2 = \begin{pmatrix} 0 & 1 \\ 0 & 0 \end{pmatrix}\:,
\end{align}
for which we have $S^- e_0=e_1$, $S^+ e_1=e_0$ and $S^+ e_0=S^- e_1=0$.

For the anisotropy parameter $\Delta$, we will generally choose 
 \begin{equation} \Delta>1,
 \end{equation}
  corresponding to the {\it Ising phase}, which assures that both terms in the second line of \eqref{eq:twosite} are non-negative.
 Now, the XXZ Hamiltonian on $\mathcal{G}$ is given by
\begin{equation} \label{XXZHam}
H=H_{\mathcal{G}} = \sum_{\{x,y\}\in \mathcal{E}} h_{xy}\:,
\end{equation}
with the understanding that $h_{xy}$ acts as the identity on the sites different from $x$ and $y$. Since $\mathcal{G}$ is finite, equation \eqref{XXZHam} rigorously defines $H_\mathcal{G}$ as a self-adjoint operator in $\mathcal{H}_{\mathcal{G}}$.
Let us call 
\begin{equation} 
\Omega:=\bigotimes_{x\in\mathcal{V}} e_0 
\end{equation}
the ``all-spins-up" or ``vacuum" state. In view of this terminology, we will also use the term ``particles" instead of ``down-spins", with $S_X^-$ and $S_X^+$ taking the role of creation and annihilation operators at site $x$. For any $\varepsilon>0$ and any $x\in\mathcal{V}$, note that $\Omega$ is --- up to a complex phase --- the unique normalized element of $\ker\left(H_\mathcal{G}+\varepsilon \mathcal{N}_x\right)$. (Using connectedness of $\mathcal{G}$ and that the XXZ system in the Ising phase is ``frustration free'', this is seen as in the case of the chain, compare \cite{FS}. It also follows from our methods below, e.g.\ \eqref{eq:dropestimate}.) Here we have defined the {\it local particle number operator} at the vertex $x$ as 
\begin{equation}
\mathcal{N}_x:= S_x^- S_x^+ =\begin{pmatrix} 0 & 0 \\ 0 & 1 \end{pmatrix}_x = \frac{1}{2} \idty - S_x^3\:.
\end{equation}

Next, for any subset $X\subset \mathcal{V}$, we define
\begin{equation} \label{eq:spinbasis}
\phi_X:=\left(\prod_{x\in X}S_x^-\right) \Omega\:,
\end{equation}
which corresponds to the state of down-spins located at the sites in $X$ and up-spins everywhere else. In particular, we have that $\{\phi_X\}_{X\subset\mathcal{V}}$ is an orthonormal basis of $\mathcal{H}_\mathcal{G}$.
We get 
\begin{equation} 
\mathcal{N}_x\phi_X=\begin{cases} \phi_X\qquad&\text{if}\qquad x\in{X}\\0\qquad&\text{if}\qquad x\notin{X}\:.\end{cases}
\end{equation}
It is now important to note that $H_{\mathcal{G}}$ conserves the total number of down-spins or particles, i.e., if we define
\begin{equation} 
\mathcal{N}:=\sum_{x\in\mathcal{G}}\mathcal{N}_x\:,
\end{equation} 
this means that $[H_\mathcal{G},\mathcal{N}]=0$, as seen by a calculation. Let  $\mathcal{H}_{\mathcal{G}}^N$ denote the eigenspace corresponding to the eigenvalue $N\in\sigma(\mathcal{N})=\{0,1,\dots,|\mathcal{V}|\}$. It is not hard to see that
\begin{equation}
\mathcal{H}_\mathcal{G}^N=\spann\{\phi_X: |X|=N\}\,.
\end{equation}
Moreover, we define $H^N:=H\upharpoonright_{\mathcal{H}_\mathcal{G}^N}$ and note that 
\begin{equation} \label{Hdecomp} 
H=\bigoplus_{N=0}^{|\mathcal{V}|}H^N
\end{equation}
as an operator on $\mathcal{H}_\mathcal{G}=\bigoplus_{N=0}^{|\mathcal{V}|}\mathcal{H}_\mathcal{G}^N$.

\subsection{Equivalence to a Schr\"odinger type operator}\label{subsec:equiv}

Here we still assume, as in Section~\ref{XXZfinite}, that $\mathcal{G}$ is finite, i.e.\ that $|\mathcal{V}|<\infty$. The purpose of this section is to show the equivalence of $H^N$ to a Schr\"odinger operator acting on $N$ hard-core bosons with the sites in $\mathcal{V}$ as their configuration space. Here we view down-spins as particles embedded in a ``vacuum'' of up-spins.  The vertices of the $N$-th symmetric product ${\mathcal G}_N$ (i.e.\ the $N$-element subsets of  $\mathcal{V}$) are the different possible configurations of $N$ particles occupying $N$ sites in $\mathcal{V}$. This corresponds to hardcore bosons as no site can be occupied by more than one particle. Thus the symmetric product graphs of $\mathcal{G}$ become a tool to more explicitly describe the operators $H^N$. 

To be more explicit, let us thus consider the Hilbert space $\ell^2(\mathcal{V}_N)$
with inner product $\langle\cdot,\cdot\rangle$ given by $\langle f,g\rangle=\sum_{X\in\mathcal{V}_N}\overline{f(X)}g(X)$. 
For any $X\in\mathcal{V}_N$, we now define the functions $\delta_X:\mathcal{V}_N\rightarrow\C$ to be $\delta_X(Y)=\delta_{XY}$. In particular, we also get $\langle\delta_X,\delta_Y\rangle=\delta_{XY}$.

We will from now on make the natural identification of the canonical basis vectors $\delta_X$ with the $N$-down-spin vectors $\phi_X$ from \eqref{eq:spinbasis} and thus of the Hilbert space $\ell^2({\mathcal V}_N)$ with $\mathcal{H}_{\mathcal{G}}^N$. We now proceed to showing that with this identification $H^N$ takes the form of a discrete Schr\"odinger operator on $\ell^2(\mathcal{V}_N)$. To this end, we introduce the following operators:
\begin{definition}
Let $\mathcal{G}=(\mathcal{V,E})$ be a finite graph. For any $0\leq N\leq |\mathcal{V}|$, let $D_N$ be the multiplication operator by the surface area function $S(X)$ in $\ell^2(\mathcal{V}_N)$, i.e.
\begin{align} \label{eq:defDN}
(D_Nf)(X)&=S(X)f(X)=|\partial X|f(X)\:
\end{align}
for $X\in {\mathcal V}_N$. Moreover, let $A_N$ denote the adjacency operator on $\mathcal{G}_N$, i.e., for all $X\in \mathcal{V}_N$,
\begin{equation}
(A_Nf)(X) = \sum_{Y\in {\mathcal V}_N:\{X,Y\} \in \mathcal{E}_N} f(Y) = \sum_{Y\in {\mathcal V}_N: X \Delta Y \in \mathcal{E}} f(Y).
\end{equation}
\end{definition}

We use the letter ``$D$'' in \eqref{eq:defDN} because, by Remark~\ref{neighbor}, $D_N$ is the degree function of ${\mathcal G}_N$.

In describing the restriction of $H_{\mathcal G}$ to $\ell^2(\mathcal{V}_N)$, we will find that the contributions of the local operators $h_{xy}$ are non-vanishing on a state $\phi_X$ only when $\{x,y\}\in\partial X$, i.e.\ only at the surface of a configuration $X$. In this case, $h_{xy}$ describes a spin or particle hopping operation (due to the XX-term in \eqref{eq:twosite}) as well as a potential contribution (corresponding to the $Z$-term).

\begin{proposition} \label{prop:HNrep}
For every $N\in \N$ and $\Delta>1$,
\begin{equation} \label{HNrep}
H^N = -\frac{1}{2\Delta} A_N +\frac{1}{2}D_N\:.
\end{equation} 
\end{proposition} 

Note that, as the graph Laplacian $({\mathcal L}_Nf)(X)= \sum_{Y\in {\mathcal V}_N: X \Delta Y \in {\mathcal E}} (f(X)-f(Y))$ on ${\mathcal G}_N$ is related to the adjacency operator $A_N$ by $\mathcal{L}_N=A_N-D_N$, we can re-express \eqref{HNrep} as
\begin{equation} \label{HNrep2}
H^N =-\frac{1}{2\Delta}\mathcal{L}_N+\frac{1}{2}\left(1-\frac{1}{\Delta}\right) D_N\:.
\end{equation} 
Also note that $H^0=0$ on the one-dimensional subspace $\ell^2({\mathcal V}_0)$ (representing the all-spins-up vector).

\begin{proof}(of Proposition~\ref{prop:HNrep})
Identifying $\phi_X$ with $\delta_X$ as above, \eqref{HNrep} amounts to showing that 
\begin{equation} \label{HNrep3}
H \phi_X = -\frac{1}{2\Delta} \sum_{Y\in {\mathcal V}_N: X\Delta Y \in {\mathcal E}} \phi_Y +\frac{1}{2}S(X) \phi_X
\end{equation}
for any $X\in\mathcal{V}_N$. The main ingredient to proving this is the identity
\begin{equation} h_{xy}=-\mathcal{N}_x\mathcal{N}_{y}+\frac{1}{2}\mathcal{N}_x+\frac{1}{2}\mathcal{N}_{y}-\frac{1}{2\Delta}\left(S_x^-S_{y}^++S_{y}^-S_x^+\right)\:,
\end{equation}
which can be checked by an explicit calculation. 

Now, for any $X\subset\mathcal{V}$ with $|X|=N$, there are three possibilities:
\begin{itemize}
\item First possibility: $x\in X$ and $y\in X$. In this case, we get 
\begin{equation}
-\mathcal{N}_x\mathcal{N}_y\phi_X+\frac{1}{2}\mathcal{N}_x\phi_X+\frac{1}{2}\mathcal{N}_y\phi_X=-\phi_X+\frac{1}{2}\phi_X+\frac{1}{2}\phi_X=0
\end{equation}
as well as $S_y^+ S_x^-\phi_X=0$ and $S_x^+ S_y^-\phi_X=0$, 
from which we conclude that $h_{xy}\phi_X=0$.
\item Second possibility: $x\notin X$ and $y\notin X$. In this case, we get 
\begin{equation}
-\mathcal{N}_x\mathcal{N}_y\phi_X+\frac{1}{2}\mathcal{N}_x\phi_X+\frac{1}{2}\mathcal{N}_y\phi_X=0+0+0=0
\end{equation}
as well as $S_y^- S_x^+\phi_X=0$ and $S_x^- S_y^+\phi_X=0$. Thus, once again, $h_{xy}\phi_X=0$.
\item Third possibility: One vertex is an element of $X$, while the other one is not. Without loss let $x\in X$ but $y\notin X$. In this case we get 
\begin{equation}
-\mathcal{N}_x\mathcal{N}_y\phi_X+\frac{1}{2}\mathcal{N}_x\phi_X+\frac{1}{2}\mathcal{N}_y\phi_X=\frac{1}{2}\phi_X
\end{equation}
as well as $S_{y}^+S_x^-\phi_X=0$ and $S_{y}^-S_x^+\phi_X=\phi_{(X\setminus\{x\})\cup\{y\}}$.
We conclude that 
\begin{equation} 
h_{xy}\phi_X=-\frac{1}{2\Delta}\phi_{(X\setminus\{x\})\cup\{y\}}+\frac{1}{2}\phi_X \,.
\end{equation}
\end{itemize}
Now, for any $X\in\mathcal{V}_N$ consider
\begin{equation}\label{eq:equivalence}
H \phi_X=\sum_{\{x,y\}\in\mathcal{E}}h_{xy}\phi_X=\sum_{x\in X}\left(\sum_{y\notin X: \{x,y\}\in\mathcal{E}}\left(-\frac{1}{2\Delta}\phi_{(X\setminus\{x\})\cup\{y\}}+\frac{1}{2}\phi_X\right)\right)\:,
\end{equation}
since only edges $\{x,y\}\in\mathcal{E}$ for which the third possibility holds contribute to the sum. 

It follows from Remark~\ref{neighbor} that the set of all $Y\in\mathcal{V}_N$ such that $X\triangle Y\in\mathcal{E}$ is given by
\begin{equation} \label{eq:adjsites}
\{Y:X\triangle Y\in\mathcal{E}\}=\bigcup_{x\in X}\left(\bigcup_{y\notin X: \{x,y\}\in\mathcal{E}} \left((X\setminus\{x\})\cup\{y\}\right)\right)\:.
\end{equation}
This yields
\begin{equation} 
\sum_{x\in X}\left(\sum_{y\notin X: \{x,y\}\in\mathcal{E}}\phi_{(X\setminus\{x\})\cup\{y\}}\right)=\sum_{Y\in\mathcal{V}_N:X\triangle Y\in\mathcal{E}}\phi_Y\:.
\end{equation}
Moreover, by \eqref{eq:S(X)},
\begin{equation} 
\sum_{x\in X}\left(\sum_{y\notin X: \{x,y\}\in\mathcal{E}}\phi_X\right)=\sum_{x\in X}\left(\sum_{y\notin X: \{x,y\}\in\mathcal{E}}1\right)\phi_X=S(X)\phi_X\:.
\end{equation}
These two identities plugged back into \eqref{eq:equivalence} show \eqref{HNrep3} and thus the proposition.
\end{proof}

\begin{remark} \label{rem:extfield} {\rm Particle number conservation of the XXZ system is preserved under the addition of an external magnetic field in the $3$-direction. In particular, we will consider the Hamiltonian 
\begin{equation}
H + \sum_{x\in \mathcal{V}} v(x) \mathcal{N}_x
\end{equation}
with $H$ from \eqref{XXZHam} (note that the field term differs from $\sum_x v(x) S_x^3$ only by an energy shift). The corresponding $N$-particle restrictions then become $H^N + V_N$, where $V_N$ is the multiplication operator by the restriction of
\begin{equation} \label{eq:potentialdef}
V(X)=\sum_{x\in X}v(x)\:.
\end{equation}
to ${\mathcal V}_N$. This is simply the $N$-body exterior potential of a system of hardcore bosons under the single particle potential $v:{\mathcal V} \to \R$.}
\end{remark}

\subsection{The XXZ model on infinite graphs}  The most pragmatic way to introduce the XXZ model on a {\it countable infinite graph} $\mathcal G$ is to use Proposition~\ref{prop:HNrep} in conjunction with \eqref{Hdecomp} as a definition, i.e., to define
\begin{equation} \label{Hinfdef}
H = H_{\mathcal G} := \bigoplus_{N=0}^{\infty} H^N
\end{equation}
on ${\mathcal H}_{\mathcal G} := \bigoplus_{N=0}^{\infty} \ell^2({\mathcal V}_N)$, where $H^0=0$ and \begin{equation} \label{eq:HNinfinite} 
H^N = -\frac{1}{2\Delta}A_N + \frac{1}{2}D_N = -\frac{1}{2\Delta} {\mathcal L}_N + \frac{1}{2}(1-1/\Delta) D_N
\end{equation} 
is bounded (with $\|H^N\| \le (\frac{1}{2\Delta} + 1)Nd$) and self-adjoint for each $N$. On infinite graphs the direct sum $H$ is generally an unbounded self-adjoint operator.

Alternatively, when properly interpreted, we may also introduce an infinite XXZ system directly via $H = \sum_{\{x,y\} \in {\mathcal E}} h_{xy}$ as in \eqref{XXZHam}. The sum is initially only well-defined on ${\mathcal H}_{\mathcal{G},0}$, the incomplete inner product space given by finite linear combinations of the orthonormal vectors $\varphi_X$, $X\subset {\mathcal V}$, $|X|<\infty$. By the argument in the proof of Proposition~\ref{prop:HNrep}, ${\mathcal H}_{\mathcal{G},0}$ is invariant under $\sum_{\{x,y\}\in {\mathcal E}} h_{xy}$ and, again identifying $\varphi_X$ with $\delta_X$, coincides there with the restriction of $H_{\mathcal G}$ defined by \eqref{Hinfdef} to ${\mathcal H}_0 = \bigoplus_N \ell_0^2({\mathcal V}_N)$ (with $\ell_0^2({\mathcal V}_N)$ the {\it finitely} supported functions on ${\mathcal V}_N$). This operator is essentially self-adjoint in the Hilbert space completion of ${\mathcal H}_0$ and its closure is unitarily equivalent to \eqref{Hinfdef}.

In particular, this means that $H$ can be calculated via \eqref{XXZHam} on the vectors $\varphi_X$ (and their finite linear combinations) also for the case of infinite graphs.

\subsection{Classical droplets and droplet spectrum} \label{sec:droplets}

Recall that in \eqref{eq:S(X)} we have defined $S(X)$ to be the surface measure of the configuration $X\subset\mathcal{V}$, i.e.\
\begin{equation} \label{degsur}
S(X)=|\partial X|\:.
\end{equation}
This means that the ground states of the operator $D_N$ on $\ell^2(\mathcal{V}_N)$ are given by the functions $\delta_X$, where $X$ is a subset of $\mathcal{V}$ with exactly $N$ elements and smallest possible edge boundary, i.e.
\begin{equation} \label{eq:DGSenergy}
D_{N,min}:=\min\sigma\left(D_N\right)=\min \left\{|\partial X|: X\subset\mathcal{V}, |X|=N\right\}\:.
\end{equation}
We thus refer to the states $\delta_X$ with $|\partial X|=D_{N,min}$ as {\bf classical droplets}, as they have minimal surfaces. Moreover, we denote the set of all classical droplet configurations in ${\mathcal V}_N$ by
\begin{equation} \label{eq:classdroplets}
{\mathcal V}_{N,1} := \{ X \in {\mathcal V}_N: \, |\partial X|= D_{N,min}\}
\end{equation}
and its complement by $\bar{\mathcal V}_{N,1}:=\{X\in\mathcal{V}_N:X\notin\mathcal{V}_{N,1}\}$. 

The determination of all minimizers of \eqref{eq:DGSenergy}, i.e., of all subsets $X$ of the graph of given volume $N$ with minimal surface measure $|\partial X|$ is known as the {\it edge-isoperimetric problem} on ${\mathcal G}$, see \cite{Bezrukov} for a survey as well as further references. We include some results on this for the examples of most interest to us in Appendix~\ref{sec:AppB}.

Being in the Ising phase $\Delta>1$ means that, in ways to be made specific below, the operator $H^N$ from \eqref{eq:HNinfinite} is dominated by the potential term $D_N$ (the limiting Hamiltonian as $\Delta\to\infty$). A first consequence of considering the Ising phase is the existence of a ground state gap for the XXZ model on arbitrary {\it infinite} connected graphs $\mathcal G$: We have $D_{N,min} \ge 1$ for all $N\in \N$ (all finite subsets of $\mathcal V$ have non-empty boundary). As this holds uniformly in $N$ and $-\frac{1}{2\Delta}{\mathcal L}_N$ is non-negative, \eqref{HNrep2} shows that 
\begin{equation} \label{eq:dropestimate}
H^N \ge \frac{1}{2}\left(1-\frac{1}{\Delta}\right) 
\end{equation}
for all $N\in \N$. Thus $E_0=0$ (coming from $H^0=0$ on the one-dimensional space spanned by the all-up-spins state) is the non-degenerate ground state energy of $H$ with a ground-state gap of at least $\frac{1}{2}(1-\frac{1}{\Delta})$. In concrete examples the gap will generally be larger.

In much of the remainder of this work we ask if the Hamiltonians $H^N$, in the Ising phase $1<\Delta<\infty$ and for energies in a vicinity of the classical droplet energy $D_{N,min}/2$, show a {\it quantum droplet phase}, separated from other parts of the spectrum by an additional gap and with eigenstates corresponding to spin configurations with a single cluster of down-spins, at least approximately.

Among non-droplet configurations the potential $D_N$ is bounded below by 
\begin{equation}
\bar{D}_N := \min \{S(X):\, X \in \bar{\mathcal V}_{N,1}\}.
\end{equation}
Obviously, $\bar{D}_N \ge D_{N,\min} + 1$, but the difference can be larger. For example, if ${\mathcal G} = ({\mathcal V}, {\mathcal E})$ is regular, i.e., each $x\in {\mathcal V}$ has a fixed number $d$ of next neighbors, then $\bar{D}_N \ge  D_{N,\min} + 2$ (in this case $d_{tot}(X) = Nd$ in \eqref{eq:georel} is a constant, so that each $S(X)$ differs by an even integer from $D_{N,min}$). 

This and the second form of the Hamiltonian $H^N$ in \eqref{eq:HNinfinite} motivates us to define the {\bf droplet spectrum} of $H^N$ as 
\begin{equation} \label{dropband} 
\delta_N := \left(0,\frac{1}{2}\left(1-\frac{1}{\Delta}\right) \bar{D}_N\right) \cap \sigma(H^N).
\end{equation}
The droplet spectrum excludes the vacuum energy $E_0=0$ and (again due to non-negativity of $-\frac{1}{2\Delta}{\mathcal L}_N$) energies associated with non-droplet configurations: 
\begin{equation} 
\bar{P}_N H_N \bar{P}_N \ge \frac{1}{2}\left(1-\frac{1}{\Delta}\right) \bar{D}_N,
\end{equation}
where $\bar{P}_N$ is the orthogonal projection onto $\ell^2(\bar{\mathcal V}_{N,1})$.

The much studied prototypical example of the above construction is the exactly solvable one-dimensional infinite XXZ chain, i.e., the case where $\mathcal G$ is the one-dimensional lattice $\Z$ (we will generally simply write $\Z^d$ for the graph with vertices $x\in \Z^d$ and edges given by next neighbors in the $\ell^1$-distance). In this case the droplet phase at the bottom of the 1D XXZ chain was first described in \cite{NS}, see also \cite{NSS} and \cite{F}. Here $\bar{D}_N=4$ for all $N$ and the droplet bands are explicitly given by the intersection of $(0,2(1-1/\Delta))$ with
\begin{equation} \label{dropband1D}
\tilde{\delta}_N = \left[ \tanh(\rho) \cdot \frac{\cosh(N\rho)-1}{\sinh(N\rho)}, \tanh(\rho) \cdot \frac{\cosh(N\rho)+1}{\sinh(N\rho)} \right], 
\end{equation}
where $\rho:=1/\Delta$. The intervals $\tilde{\delta}_N$ are nested, $\tilde{\delta}_{N+1} \subset \tilde{\delta}_N$ for all $N\in \N$, and contract (exponentially fast) to the limit point $\tilde{\delta}_{\infty} = \sqrt{1-\frac{1}{\Delta^2}}$. Their union is $\tilde{\delta}_1= [1-1/\Delta,1+1/\Delta]$. For $\Delta>3$ this is strictly contained in $(0,2(1-1/\Delta))$ and thus, in particular, the {\it union of all droplet bands} of $H$ is separated by a non-trivial gap from the higher spectrum of $H$.

In the next section we will introduce a method which allows to prove related results for XXZ systems on larger classes of graphs and does not depend on exact solvability.

\section{An abstract approach to the existence of gaps} \label{sec:abstractgap}

Let ${G}=({V,E})$ be a countable graph, write $x\sim y$ for $\{x,y\}\in E$, and assume that is has bounded degree $\sup_{x\in V} D(x)<\infty$, where $D(x)=|\{y\in V:  x\sim y\}|$. In particular, this means that the adjacency operator $A_{G}$ on ${G}$ is a bounded operator on $\ell^2({V})$. For this section, let us drop the index $G$ and just write $A=A_G$.

Consider a disjoint decomposition $V=V_1 \cup V_2$ of the vertices and let $B$ be the ``boundary hopping operator'' from $\ell^2(V_2)$ to $\ell^2(V_1)$, i.e.,
\begin{equation} \label{hop21}
(Bf)(x) := \sum_{y\in V_2, y \sim x} f(y)
\end{equation}
for all $f\in \ell^2(V_2)$ and $x\in V_1$. Note that its adjoint is the boundary hopping operator from $\ell^2(V_1)$ to $\ell^2(V_2)$,
\begin{equation} \label{hop12}
(B^*g)(y) = \sum_{x\in V_1, x \sim y} g(x)
\end{equation}
for all $g\in \ell^2(V_1)$ and $y\in V_2$. 

These operators are bounded. More precisely:

\begin{lemma} \label{boundaryhop}
Let
\begin{equation} 
d_1 := \sup_{x\in V_1} |\{y\in V_2: x \sim y\}| \quad \mbox{and} \quad d_2 := \sup_{y\in V_2} |\{x\in V_1: x \sim y\}|. \end{equation}
Then
\begin{equation} \label{hopbound}
 \|B\| = \|B^*\| \le \sqrt{d_1 d_2}. 
 \end{equation}
\end{lemma}

\begin{proof} We have
\begin{equation}
\|B f\|^2 =  \sum_{x\in V_1} \Big| \sum_{y\in V_2, y\sim x} f(y) \Big|^2 
 \le  d_1 \sum_{x\in V_1} \sum_{y\in V_2, y\sim x} |f(y)|^2, 
\end{equation}
where we used that the second sum has at most $d_1$ non-zero terms. We exchange the order of summation and further bound by
\begin{equation}
d_1 \sum_{y\in V_2} \sum_{x\in V_1, x \sim y} |f(y)|^2 \le d_1 d_2 \|f\|^2,
\end{equation}
as the second sum now has at most $d_2$ non-zero terms. Thus we have shown \eqref{hopbound}.
\end{proof}

Now assume that $U$ is the multiplication operator on $\ell^2(V)$ by a bounded function $U:V\to \R$ and denote $\sigma_j := \overline{\{U(x):x\in V_j\}}$, i.e., the spectra of the restrictions $U_j$ of $U$ to $\ell^2(V_j)$, $j=1,2$. We assume that 
\begin{equation} \label{eq:gap2}
\min \sigma_1 = \inf_{x\in V_1} U(x) \ge \sup_{x\in V_2} + S = \max \sigma_2+S,
\end{equation} 
i.e., $\sigma_1$ lies above $\sigma_2$ and there is a gap of size $S>0$ between $\sigma_1$ and $\sigma_2$. Due to the boundedness of $A$, for sufficiently small $g>0$ this leads to a gap in the spectrum of the Schr\"odinger-type operator
\begin{equation} \label{Schroetype}
H = -gA +U
\end{equation}
on $\ell^2(V)$ and we want to prove lower bounds on the size of this gap, which are better (in our applications) than the rough bound \begin{equation} \label{rough}
S-2g\|A\|.
\end{equation}

Let $E\in \R$ and write $H-E$ as a block operator with respect to the decomposition $\ell^2(V) = \ell^2(V_1) \oplus \ell^2(V_2)$ of the Hilbert space,
\begin{equation} \label{blockdec}
H-E = \begin{pmatrix} C_1 & gB \\ gB^* & C_2 \end{pmatrix}.
\end{equation}
Here $B$ and $B^*$ are the boundary hopping operators between $\ell^2(V_2)$ and $\ell^2(V_1)$ and $C_j = -gA_j + U_j -E$, where $A_j$ are the adjacency operators on the restricted graphs $G_j = (V_j, E_j=\{\{x,y\}\in E: x, y \in V_j\})$, $j=1,2$. Let $D$ by the multiplication operator on $\ell^2(V)$ by the degree function $D(x)$.

There is now a subtle distinction to be made between $D_i$ and $D\upharpoonright_{\ell^2(V_i)}$. Here, $D_i$ is the operator of multiplication by the degree function of the restricted graph $G_i=(V_i,E_i)$, where $E_i=\{\{x,y\}\in E: x,y\in V_i\}$. Since for any $x\in V_i$, we get
\begin{equation}
D(x)=|\{y\in V_i: x\sim y\}|+|\{y\in V\setminus V_i:x\sim y\}|\geq |\{y\in V_i: x\sim y\}|=D_i(x)\:,
\end{equation}
we conclude 
\begin{equation} \label{eq:degreecompression}
D\upharpoonright_{\ell^2(V_i)}\geq D_i\:.
\end{equation}
Also, since the negative Laplacian $-\mathcal{L}_i=D_i-A_i$ on $G_i$ is non-negative, we get
\begin{equation} \label{eq:laplestim}
-A_i\geq -D_i\:.
\end{equation}
Now, assume that 
\begin{equation} \label{eq:assE1}
E > \sup_{x\in V_2} U(x) +g\|A_2\|.
\end{equation} 
so that $C_2$ is invertible. Assume that the Schur complement $S:=C_1 - g^2 B C_2^{-1} B^*$ of the block operator \eqref{blockdec} is invertible as well. Then $H-E$ is invertible and
\begin{equation} \label{eq:SchurInv}
(H-E)^{-1} = \begin{pmatrix} S^{-1} & -gS^{-1}BC_2^{-1} \\ -gC_2^{-1} B^* S^{-1} & C_2^{-1} (I + g^2 B^* S^{-1}B) C_2^{-1} \end{pmatrix}.
\end{equation}

To prove that $S$ is invertible, now assume in addition that 
\begin{equation} \label{eq:assE2} 
E<\inf_{x\in V_1} (U(x)-gD(x))
\end{equation} 
so that, by \eqref{eq:laplestim} and \ref{eq:degreecompression},
\begin{align}
C_1 &= -gA_1 +U_1-E=-gA_1 +U_1-E \ge -gD_1 +U_1-E\\ &\geq -gD\upharpoonright_{\ell^2(V_1)} +U_1-E\geq \inf_{x\in V_1} (U(x)-gD(x)) -E >0 \notag
\end{align}
and therefore $C_1$ is invertible. Thus invertibility of the Schur complement $S=C_1-g^2 BC_2^{-1}B^* = C_1(I-g^2 C_1^{-1}B C_2^{-1} B^*)$ follows if
\begin{equation}
g^2 \|C_1^{-1} B C_2^{-1} B^*\| < 1.
\end{equation}
From Lemma~\ref{boundaryhop} and the considerations above we get
\begin{align}
\|C_1^{-1} B C_2^{-1} B^*\| & \le  d_1 d_2 \|C_1\|^{-1} \|C_2\|^{-1} \\
& \le  \frac{d_1 d_2}{\inf_{x\in V_1} (U(x)-gD(x)-E) \cdot (E-\sup_{x\in V_2} U(x) -g\|A_2\|)}. \notag
\end{align}

In summary, we have shown most of

\begin{proposition} \label{gapcrit1} Assume \eqref{eq:gap2} and that $E\in \R$ satisfies \eqref{eq:assE1} and \eqref{eq:assE2}. Let $g>0$ be such that
\begin{equation} \label{gapcrit}
g^2 < \frac{\inf_{x\in V_1} (U(x)-gD(x)-E) \cdot (E-\sup_{x\in V_2} U(x) -g\|A_2\|)}{d_1 d_2}\;.
\end{equation}
Then $E$ is not in the spectrum of $H = -gA +U$ and the Hilbert space dimension of the spectral projection of $H$ onto $(-\infty,E)$ is at least $|V_2|$.
\end{proposition}

We only need to comment on the concluding statement on dimensions. For all $\varphi\in \ell^2(V_2)\setminus \{0\}$, we get from the representation \eqref{blockdec} that
\begin{eqnarray}
\left\langle \begin{pmatrix} 0 \\ \varphi \end{pmatrix} , (H-E) \begin{pmatrix} 0 \\ \varphi \end{pmatrix} \right\rangle & = & \langle \varphi, C_2 \varphi \rangle = \langle \varphi, (-gA_2+U_2-E)\varphi \rangle \\
& \le & (\sup_{x\in V_2} U(x) + g\|A_2\| - E) \|\varphi\|^2 <0 \notag
\end{eqnarray}
by \eqref{eq:assE1}. Thus the claim follows from the variational principle. 

\vspace{.3cm}

In our applications below we consider the family of Schr\"odinger-type operators $H^N$ on the graphs ${\mathcal G}_N$ from \eqref{eq:HNinfinite} above (for a given graph $\mathcal G$), as well as their restrictions ${\mathcal G}_{N,1}$ and ${\mathcal G}_{N,2}$ to ${\mathcal V}_{N,1}$ and $\bar{\mathcal V}_{N,1}$, the droplet and non-droplet configurations, respectively. We aim at showing the existence of a droplet regime at the bottom of the spectrum of the $H^N$ for suitable types of graphs $\mathcal G$, including the presence of a gap which separates the droplet spectrum from higher spectral bands. Specifically, we are interested in large particle numbers $N$, and ideally in results which hold uniformly in the particle number (and thus also for the spectrum of $H$, i.e., the union of the spectra of all $H^N$). The trivial bound \eqref{rough} on the spectral gap will generally be useless here, because the norms $\|A_N\|$ of the adjacency operators are unbounded in $N$, even for one-dimensional graphs such as $\Z$.

The bound \eqref{gapcrit} is more suitable, because (i) in our application we have $U=\frac{1}{2}D$ and thus the first factor in the numerator of \eqref{gapcrit} becomes $\frac{1}{2}(1-\frac{1}{\Delta}) \bar{D}_N-E$ which generally grows with $N$ if $\Delta>1$, (ii) the norm of the adjacency operator restricted to the droplet configurations is bounded in $N$, at least for ``one-dimensional'' graphs, allowing to control the $N$-dependence of the second factor in the numerator, and (iii) the ``boundary hopping degrees'' $d_1$ and $d_2$, again in one-dimensional situations, are uniformly bounded above in $N$ as well.

In the next section we make this more precise for some examples. In particular, we demonstrate that strips of arbitrary width indeed show a droplet regime at the bottom of the spectrum if $\Delta$ is sufficiently large, depending on the width of the strip. We will also discuss that for a two-dimensional XXZ system there is no droplet regime, at least not uniform in the particle number.

\section{Gaps in the spectra of the XXZ Hamiltonian on various graphs} \label{sec:examples}

\subsection{Gaps in the spectrum of the 1-d chain} \label{gapchain}

In the following, we want to use Proposition \ref{gapcrit1} to show existence of spectral gaps in higher order bands of the XXZ spin chain -- uniformly in the particle number. Hence, the graph $\mathcal{G}=(\mathcal{V,E})$ for which we define the full XXZ Hamiltonian $H_\mathcal{G}$ as in \eqref{Hinfdef} is given by $\mathcal{V}=\Z$ and $\mathcal{E}=\{\{x,y\}\subset \Z : |x-y|=1\}$.

To be more precise, we are interested in finding spectral gaps between the {\it cluster bands} $\mathcal{C}(k)$ given by
\begin{equation}
\mathcal{C}(k)=\left[k\left(1-\frac{1}{\Delta}\right),k\left(1+\frac{1}{\Delta}\right)\right]\:,
\end{equation}
where $k\in\N$. (In addition, we define $\mathcal{C}(0):=\{0\}$.) In \cite{FS}, it was shown that 
\begin{equation}
\bigcup_{k\geq 0}\mathcal{C}(k)\subset\sigma(H_\mathcal{G})
\end{equation}
and moreover that for any $N\in\N_0$:
\begin{equation}
\label{eq:maxspecN}
\sup\sigma\left(H_\mathcal{G}^N\right)=N\left(1+\frac{1}{\Delta}\right)\:,
\end{equation}
where $H_\mathcal{G}^N$ is the operator on $\ell^2(\mathcal{V}_N)$ as defined in \eqref{eq:HNinfinite}.

In order to prove the existence of gaps between $\mathcal{C}(k)$ and $\mathcal{C}(k+1)$ in the spectrum of $H_\mathcal{G}$, we will show under certain assumptions on $\Delta$ that for any $N\in\N_0$, there is a gap $\gamma_k(N)$ in the spectrum of $H_\mathcal{G}^N$  between $\mathcal{C}(k)$ and $\mathcal{C}(k+1)$, i.e.
\begin{equation} 
\rho(H_\mathcal{G}^N)\cap\left(k\left(1+\frac{1}{\Delta}\right),(k+1)\left(1-\frac{1}{\Delta}\right)\right)=:\gamma_k(N)\neq \emptyset
\end{equation}
and moreover that 
\begin{equation} \label{eq:unifgap}
\bigcap_{N\geq 0}\gamma_k(N)=:\gamma_k\neq\emptyset\:,
\end{equation}
which means that there is a gap $\gamma_k$ between $\mathcal{C}(k)$ and $\mathcal{C}(k+1)$ in the spectrum of $H_\mathcal{G}$. The case $k=0$ is the well-known existence of a ground state gap for $\Delta>1$,  
\begin{equation} 
\rho(H_\mathcal{G})\cap \left(0,1-\frac{1}{\Delta}\right)=\left(0,1-\frac{1}{\Delta}\right),
\end{equation}
while the case $k=2$ was covered in \cite{FS} by showing that for $\Delta>3$ we have
\begin{equation} 
\rho(H_\mathcal{G})\cap \left(1+\frac{1}{\Delta},2\left(1-\frac{1}{\Delta}\right)\right)=\left(1+\frac{1}{\Delta},2\left(1-\frac{1}{\Delta}\right)\right)\:.
\end{equation}
Thus we will restrict ourselves to the cases $k\geq 2$ from now on.

In view of \eqref{eq:maxspecN}, note that it suffices to only consider the operators $H_\mathcal{G}^N$, where $N\geq k+1$ and show the existence of spectral gaps $\gamma_k(N)$ such that \eqref{eq:unifgap} holds. Also, note that a priori, we have to assume that $\Delta>2k+1$ to ensure that $\sup\:\mathcal{C}(k)<\inf\:\mathcal{C}(k+1)$.

Now fix any $N\geq k+1$ and consider the operator $H_\mathcal{G}^N$ on $\ell^2(\mathcal{V}_N)$. Let us split $\mathcal{V}_N$ into the disjoint union $\mathcal{V}_N=V_1\cup V_2$, where we define $V_2$ to be the set of those configurations $X_2\in V_2$ such that $S(X_2)\leq 2k$. In other words, $\ell^2(V_2)$ is the spectral subspace for which $D_N\leq 2k$ holds. 
The set $V_1\subset\mathcal{V}_N$ is then just defined to be $V_1=\mathcal{V}_N\setminus V_2$ and we have $D_N(X_1)\geq 2k+2$ for any $X_1\in V_1$. Since we are considering the XXZ model, we get $g=\frac{1}{2\Delta}$ and $U=\frac{1}{2}D_N$ for our application of Proposition \ref{gapcrit1}. 

In order to avoid our arguments being too technical, let us also introduce the following terminology:

(i) Let $X\in\mathcal{V}_N$, which means that $X$ is a subset of $\Z$ with exactly $N$ elements. We call any $x\in X$ a {\it particle} and any sequence of consecutive numbers $\{x,x+1,\dots,x+n-1\}\subset X$ an {\it $n$-particle cluster} ($n\in\{1,\dots,N\}$).
In this terminology, the interpretation of $V_2$ is that of the set of all configurations with $k$ or less clusters. This in turn implies that $V_1$ is the set of configurations with $k+1$ or more clusters. Also, note that for any $X\in\mathcal{V}_N$, we have that $\frac{1}{2}S(X)$ is the number of clusters this configuration consists of.

(ii) Given $X,Y\in\mathcal{V}_N$, such that $\{X,Y\}\in\mathcal{E}_N$ or equivalently such that $X\triangle Y\in\mathcal{E}$, this means that there exists a $x\in X\cap Y^c$ and $y\in X^c\cap Y$ such that $Y=(X\setminus\{x\})\cup\{y\}$ (cf.\ Remark~\ref{neighbor}). In this case, we say that the configuration $Y$ is obtained from $X$ through {\it hopping} of the particle at $x$ from $x$ to $y$. In general, given any $X\in\mathcal{V}_N$, when we speak of {\it hopping}, we mean the adjacent sites $\{Y\in\mathcal{V}_N: X\triangle Y\in\mathcal{E}\}$ and think of them as obtained through the different possibilities that particles in $X$ can hop from sites $x\in X$ to $y\in Y\setminus X$. 

Let us now estimate the quantities we need in order to apply this result. 

\begin{itemize}

\item We begin with finding an upper bound for $d_1d_2$. 

Observe that $X\in V_2$, i.e., a configuration with at most $k$ connected clusters, has at most $2k$ neighbors in ${\mathcal G}_N$ with $k+1$ clusters (and no neighbors with more than $k+1$ clusters). The extreme case $2k$ happens if there are $k$ connected clusters, each consisting of at least two particles, and the spatial separation between all clusters is at least $2$. In this case each endpoint of each cluster can hop to exactly one neighbor outside $X$, splitting one cluster into two. Thus $d_1\le 2k$.

On the other hand, observe that hopping can decrease the number of clusters by at most one. It thus suffices to only consider configurations of $k+1$ clusters when trying to find an estimate for $d_2$. 
Now, a decrease of the number of clusters from $k+1$ to $k$ will only happen if there is a one-particle cluster at distance two from another cluster. The extremal case of this situation is given by $k+1$ single particles lined up in a chain with distance two between each pair. Then, the very left and the very right particle each have one way of hopping such that it forms a two-particle cluster with its nearest neighboring particle, while the $k-1$ particles inside have two ways each for hopping and connecting. Altogether, this yields $d_2 \le 2(k-1)+2=2k$. We therefore conclude 
\begin{equation} \label{eq:d1d2}
d_1d_2\leq 4k^2\:.
\end{equation} 

\item Next, we can estimate $ \|A_2\|\leq \|D_2\|\le 2k$, as it was already argued above that configurations $V_2$ can have at most $2k$ next neighbors.

\item Since $U=\frac{1}{2}D_N$ and $g=\frac{1}{2\Delta}$, it is easy to describe 
\begin{equation} \label{eq:classicclusterspec}
\sigma(U_2)=\sigma\left(\frac{1}{2}D_N\upharpoonright_{\ell^2(V_2)}\right)=\{1,2,\dots,k-1,k\}\:,
\end{equation}
since for each configuration $X\in V_2$ with $\ell\leq k$ clusters, we have $D_N(X)=|\partial X|=2\ell$.
\item Finally, we will use that 
\begin{equation}\label{eq:specestim}
\inf_{X\in V_1}\: (U(X)-gD(X))=\inf_{X\in V_1}\:\left(\frac{1}{2}\left(1-\frac{1}{\Delta}\right)D_{N}(X)\right)=(k+1)\left(1-\frac{1}{\Delta}\right)\:,
\end{equation}
which follows immediately from the fact that configurations $X\in V_1$ have $k+1\geq \ell$ clusters again with $|\partial X|=2\ell$. The spectral minimum of $D_N\upharpoonright_{\ell^2}V_1$ is thus attained by configurations with $k+1$ clusters. 

\end{itemize}
Plugging all those quantities into Equation \eqref{gapcrit} then yields that all numbers $E\in (k(1+1/\Delta), (k+1)(1-1/\Delta))$ with the additional property that
\begin{align}
\frac{1}{\Delta^2}<\frac{\left((k+1)\left(1-\frac{1}{\Delta}\right)-E\right)\left(E-k\left(1+\frac{1}{\Delta}\right)\right)}{k^2}
\end{align}
are in the resolvent set of $H_\mathcal{G}^N$. Note that this is uniform in $N$.
Let us now define 
\begin{equation} \label{eq:gapboundary}
\gamma_{\pm}(\Delta,k):=k+\frac{1}{2}\left(1-\frac{1}{\Delta}\right)\pm\frac{1}{2\Delta}\sqrt{(\Delta-1)(\Delta-(4k+1))}
\end{equation}
By an easy calculation, we then get 
\begin{proposition} \label{prop:XXZresolvent}
If $\Delta>4k+1$, then the set 
\begin{equation} \label{eq:XXZresolvent}
\left(\gamma_-(k,\Delta),\gamma_+(k,\Delta)\right)
\end{equation}
is in the resolvent set of the $XXZ$ spin chain Hamiltonian $H_{\mathcal G}$.
\end{proposition}

Note that, for fixed $k$, the intervals $(k(1+1/\Delta),(k+1)(1-1/\Delta))$ and $(\gamma_-(k,\Delta), \gamma_+(k,\Delta))$ have the same center and their endpoints coincide up to $C(k)/\Delta^2$ as $\Delta \to \infty$.


\subsection{The strip of width $M$} \label{sec:strip}

For any $M\in\{2,3,\ldots\}$, let $\mathcal{V}=\Z\times\{1,2,\ldots,M\}$ and let $\mathcal{G}=(\mathcal{V},\mathcal{E})$ be the strip of width $M$ with open boundary conditions, i.e. 
\begin{equation} 
\mathcal{E}=\left\{\{(z_1,m_1),(z_2,m_2)\}: z_1,z_2\in \Z, m_1,m_2\in\{1,2,\dots,M\}, |z_1-z_2|+|m_1-m_2|=1\right\}\:.
\end{equation} 
Moreover, for any $N\in\N_0$, let 
\begin{equation} 
H_N=-\frac{1}{2\Delta}\mathcal{L}_N+\frac{1}{2}\left(1-\frac{1}{\Delta}\right)D_N\:,
\end{equation}
where $\mathcal{L}_N$ is the Laplacian on $\mathcal{G}_N$

For simplicity, we will restrict our considerations to sufficiently large particle numbers that are multiples of $M$. This will ensure that the space of classical droplet configurations $\mathcal{V}_{N,1}$ can be described by rectangular configurations of the form
\begin{equation}
\label{eq:gluehwein}
\mathcal{R}_{z,\ell} := \{(y,m) \in \N \times \{1,\ldots,M\} : z\le y \le z+\ell-1\}\:.
\end{equation}

Hence, for any $M\in\N$, we define 
\begin{equation}
\mathcal{S}_M:=\{\ell M:\ell\in\N, \ell>\frac{M}{2}\}
\end{equation}
and we introduce the restricted Hamiltonian 
\begin{equation}
\widehat{H}_M=\bigoplus_{N\in\mathcal{S}_M}H_\mathcal{G}^N\:.
\end{equation}
For any such $N=\ell M \in\mathcal{S}_M$, the classical droplet configurations are given by rectangles of height $M$ and width $\ell$ (see Lemma \ref{lemma:strip} for a proof),
\begin{equation} 
\mathcal{V}_{N,1}=\{\mathcal{R}_{z,\ell}: z \in \Z\}\:.
\end{equation}
From this it is not hard to obtain the necessary quantities for an application of Proposition \ref{gapcrit1}, where for the sets $V_1$ and $V_2$ into which we decompose $\mathcal{V}_N=V_1\cup V_2$, we choose $V_2=\mathcal{V}_{N,1}$ and $V_1=\mathcal{V}_N\setminus\mathcal{V}_{N,1}$. 
\begin{itemize}
\item Firstly, note that $g=\frac{1}{2\Delta}$, $D=D_N$ and $U=\frac{1}{2}D_N$.
\item Next, from the fact that all elements of $\mathcal{V}_{N,1}$ are of the form $\mathcal{R}_{z,\ell}$, we can immediately determine the quantities $d_1$ and $d_2$. Since for any $\mathcal{R}_{z,\ell}$, there are $2M$ adjacent non-droplet configurations, we have $d_2=2M$. Moreover, since two adjacent configurations $X,Y\in\mathcal{V}_N$ can only differ by at most one element, we get $d_1=1$ and thus $d_1d_2=2M$.  
\item Since no classical droplet configuration in $\mathcal{V}_{N,1}$ is adjacent to another (note that we have assumed that $M>1$), we get $\|A_2\|=0$.
\item Next, by virtue of Lemma \ref{lemma:strip} observe that 
\begin{equation} 
\sup_{x\in V_2}U(x)=\frac{1}{2}\sup_{X\in\mathcal{V}_{N,1}}D_N(X)=M\:.
\end{equation}
\item Finally, we estimate 
\begin{align}
\inf_{x\in V_1}(U(x)-gD(x))&=\inf_{X\in\mathcal{V}_N\setminus\mathcal{V}_{N,1}}\left(\frac{1}{2}\left(1-\frac{1}{\Delta}\right)D_N(X)\right)\\&\geq \left(1-\frac{1}{\Delta}\right)(M+1)\:,
\end{align}
which trivially follows from the fact that $D_N$ is integer-valued and that for any $X\in\mathcal{V}_N\setminus\mathcal{V}_{N,1}$ we have $D_N(X) \ge 2M+2$. (The obvious lower bound $2M+1$ can be improved to $2M+2$, realized by taking one of the four corners of $\mathcal{R}_{z,\ell}$ and attaching it to one of the other three corners, but we haven't worked out a detailed proof of this.)
\end{itemize}

Let $\Delta>2M+1$. By Proposition \ref{gapcrit1}, any $E$ with $M<E<\left(1-\frac{1}{\Delta}\right)(M+1)$ that in addition satisfies
\begin{equation}
\frac{1}{4\Delta^2}<\frac{\left(\left(1-\frac{1}{\Delta}\right)(M+1)-E\right)(E-M)}{2M}
\end{equation}
is not in the spectrum of $H_N$. Similar to the considerations at the end of Section~\ref{gapchain} this shows that all $H_N$ with $N\in S_M$ (and thus also $\widehat{H}_M$) have a common gap whose endpoints coincide with those of $(M,\left(1-\frac{1}{\Delta}\right)(M+1))$ up to an error of at most $C(M)/\Delta^2$.

\subsection{The XXZ model in $\Z^2$} \label{sec:Z2}

Let us now consider $\mathcal{G}=(\mathcal{V,E})$, where $\mathcal{V}=\Z^2$ and $\mathcal{E}=\{\{x,y\}:x,y\in\Z^2, \|x-y\|_{\ell^1}=1\}$ and for simplicity restrict ourselves to particle numbers that are squares. In this case, the classical droplets are given by states $\delta_X$, where $X\subset \Z^2$ are squares, i.e.\ up to translations of the form $X=\{1,2,\dots, \sqrt{N}\}\times\{1,2,\dots,\sqrt{N}\}$ (cf.\ Lemma \ref{lemma:square} for more details). Since the edge boundary $\partial X$ of such a configuration has exactly $4\sqrt{N}$ elements, we have $D_{N,min}=4\sqrt{N}$, from which we get 
\begin{equation} 
H_\mathcal{G}^N\geq2\left(1-\frac{1}{\Delta}\right)\sqrt{N}\:.
\end{equation}
In particular, this means that $\inf\sigma(H_\mathcal{G}^N)\rightarrow\infty$ as $N\rightarrow\infty$. Hence, for this choice of $\mathcal{G}$ it is not possible that there exists a droplet band uniform in the particle number. Nevertheless, for each $N\in\{m^2:m\in\N\}$, let us apply Proposition \ref{gapcrit1} in order to show the existence of a spectral gap of $H_{\mathcal{G}}^N$ (with $N$ fixed) above the droplet band. Since for $N=1$, the operator $H_\mathcal{G}^1$ can be explicitly diagonalized with the use of Fourier transforms with 
\begin{equation} 
\sigma(H_\mathcal{G}^1)=\sigma_{ac}(H_\mathcal{G}^1)=\left[2-\frac{2}{\Delta},2+\frac{2}{\Delta}\right]\:,
\end{equation}
we will exclude this case from now on.

As before, we now need to collect a few needed quantities:
\begin{itemize}
\item For $X\in \mathcal{V}_{N,1}$, we have $d_2=4\sqrt{N}$ and $d_1=1$.
\item Since $N\geq 4$, there is are no square configurations $X$ and $X'$ such that $d_N(X,X')=1$, from which we conclude $\|A_2\|=0$.
\item Clearly, we have $\sigma(U_2)=\{2\sqrt{N}\}$.
\item The second smallest eigenvalue of $D_N$ is given by $\lambda_2=4\sqrt{N}+2$. This follows as one can easily find an explicit $N$-particle configuration with surface area $4\sqrt{N}+2$ and  from the fact that $\mathcal{G}$ is a regular graph and thus Lemma~\ref{lemma:graphhelp} implies that the values of $D_N(X)=S(X)$ have to be even integers.  We thus get
\begin{equation} 
\min\sigma\left(\frac{1}{2}\left(1-\frac{1}{\Delta}\right)D\upharpoonright_{\ell^2(V_1)}\right)=\left(1-\frac{1}{\Delta}\right)(2\sqrt{N}+1)\:.
\end{equation}
Hence, we get that $2\sqrt{N}<E<(1-\frac{1}{\Delta})(2\sqrt{N}+1)$ is not in the spectrum if it satisfies
\begin{align}
\frac{1}{4\Delta^2}<\frac{\left(\left(1-\frac{1}{\Delta}\right)(2\sqrt{N}+1-E\right)(E-2\sqrt{N})}{4\sqrt{N}}\:.
\end{align}
\end{itemize}
Again, for $\Delta$ sufficiently large this gives a gap almost as big as $(2\sqrt{N}, (1-1/\Delta)(2\sqrt{N}+1))$. However, we can not produce gaps in $H$, not even for its restriction to particle numbers which are squares, as the gap location now moves as $\sqrt{N}$. 

\section{Combes-Thomas bounds} \label{sec:CTbound}

The discussion of the existence of droplet spectrum would not be complete without also giving a corresponding description of eigenstates to energies in this part of the spectrum. What one would like to show is that these eigenfunctions are close to the ``classical'' droplet states $\varphi_X$, where $X$ is a droplet configuration, i.e., contained in ${\mathcal V}_1 = \bigcup_N {\mathcal V}_{N,1}$. We will accomplish this by providing a Combes-Thomas-type bound in Section~\ref{subCTbound}, which shows that, for given $N$, the Green function of (a suitably modified version of) $H^N$ for energies in the droplet spectrum decays exponentially in the $d_N$-distance from ${\mathcal V}_{N,1}$. This generalizes results in \cite{EKS1} where such bounds were proven for the case ${\mathcal V}=\Z$ (see also \cite{BW1} for a related Combes-Thomas bound).  A quick calculation in Section~\ref{efdecay} will then show that this implies a similar decay bound for eigenstates to droplet energies and, if the droplet spectrum is separated by a gap from higher spectral bands, also for the spectral projection onto the full droplet spectrum.

Throughout this section we will actually state and prove results which do not just apply to the droplet spectrum, but to energies below an increasing sequence of $k$-threshold energies, which correspond to ``multi-droplets'', i.e., states close to classical configurations with a fixed number $k$ of connected components of down-spins.

\subsection{Degree growth between neighbors in $\mathcal{G}_N$}

The exponential decay rates in ``standard'' Combes-Thomas bounds for Schr\"odinger-type operators on $\Z^N$ are inversely proportional to the number $2N$ of next neighbors of lattice sites, e.g.\ \cite{Kirsch}, and thus do not give the uniformity in the particle number $N$ desired here. However, as we will show next, due to \eqref{degsur} the {\it degree growth} between next neighbors in ${\mathcal G}_N$ satisfies an $N$-independent bound. This property will be critical in the proof of the Combes-Thomas bound near droplet energies in Section~\ref{subCTbound}.

\begin{lemma} \label{lem:deggrow}
Let ${\mathcal G}= ({\mathcal V}, {\mathcal E})$ be a countable, undirected and connected graph of bounded degree $d=\sup_{x\in {\mathcal V}} d(x) <\infty$. 
Let $X, Y \subset {\mathcal V}$ such that $X \Delta Y \in {\mathcal E}$. Then either $|\partial X|$ and $|\partial Y|$ are both infinite, or
\begin{equation} \label{XYcomp}
\left| |\partial X| - |\partial Y| \right| \le 2d-1,
\end{equation}
in which case also 
\begin{equation} \label{XYcomp2}
|\partial X|^{-1} \le 2d |\partial Y|^{-1}.
\end{equation}
\end{lemma}

\begin{proof}
Note that $X\Delta Y \in {\mathcal E}$ means that there exist $x_0 \in X$ and $y_0 \in Y$ such that $\{x_0,y_0\} \in {\mathcal E}$ and $X\setminus \{x_0\} = Y \setminus \{y_0\}$. In particular, $x_0 \not\in Y$ and $y_0 \not\in X$, $\partial X$ and $\partial Y$ are both non-empty (each contains $\{x_0,y_0\}$), and $|\partial X|$ and $|\partial Y|$ are either both infinite or both finite.
In the latter case, to show \eqref{XYcomp} it suffices to prove
\begin{equation} \label{YleX}
|\partial Y| \le |\partial X| +2d-1.
\end{equation} 
 Consider
\begin{equation} 
F:= \{\{u,v\}\in \mathcal{E}: \{x_0,y_0\} \cap \{u,v\} \not= \emptyset\}.
\end{equation}
Then 
\begin{equation} \label{Fcard}
|F|\le 2d-1
\end{equation}
as at most $d$ edges each contain either $x_0$ or $y_0$, where $\{x_0,y_0\}$ is counted twice.

Suppose $\{u,v\} \in \partial Y \setminus F$, i.e., without loss, $u\in Y$ and $v\in Y^c$, as well $\{u,v\} \cap \{x_0,y_0\} = \emptyset$. As $X\Delta Y \in {\mathcal E}$, this implies $u \in Y\setminus\{y_0\} \subset X$ as well as $v\in (Y \cup \{x_0\})^c \subset X^c$. We have shown $\{u,v\} \in \partial X$. Thus \eqref{YleX} follows from \eqref{Fcard}. This also gives the last claim,
\begin{equation} 
\frac{1}{|\partial Y|} = \frac{1}{|\partial X|} \frac{|\partial X|}{|\partial Y|} \le \frac{1}{|\partial X|} \frac{|\partial Y|+2d-1}{|\partial Y|} \le \frac{2d}{|\partial X|}.
\end{equation}
\end{proof}

Due to \eqref{eq:georel} and \eqref{degsur}, $||\partial X|-|\partial Y||=2d-1$ is not possible on a regular graph, so in this case one must have $||\partial X|-|\partial Y||\le 2d-2$. 

\subsection{Combes-Thomas bounds} \label{subCTbound}

Recall that the operators $H^N$ are defined by \eqref{HNrep} or \eqref{HNrep2}, for any $N\ge 1$ and countable, connected directed graph ${\mathcal G} = ({\mathcal V}, {\mathcal E})$ of degree bounded by $d$. As before, we denote the graph distance on ${\mathcal G}_N = ({\mathcal V}_N, {\mathcal E}_N)$ by $d_N$, see \eqref{distN}. For $X\in {\mathcal V}_N$ and subsets ${\mathcal A}, {\mathcal B} \subset {\mathcal V}_N$ we define
\begin{equation}
d_N({\mathcal A},X) = \min_{Y\in {\mathcal A}} d_N(Y,X) \:\: \mbox{and} \:\:d_N({\mathcal A}, {\mathcal B}) = \min_{X\in {\mathcal A}, Y \in {\mathcal B}} d_N(X,Y).
\end{equation}
As $d_N$ is an integer-valued metric, it is clear that both minima exist and that
\begin{equation} \label{triangle}
|d_N({\mathcal A},X) - d_N({\mathcal B},Y)| \le d_N(X,Y).
\end{equation}

We will now state and prove the Combes-Thomas bound for energies in the droplet spectrum. In fact, our result also applies to higher energies, as long as we assume an upper bound of the form $D_{N,min}+k$ on the corresponding classical energies (i.e., of the operators $D_N$). Here 
\begin{equation} 
D_{N,min} := \min _{X\in \mathcal{V}_N} D(X)
\end{equation} 
and $k\in \N_0$. Let
\begin{equation}
\mathcal{V}_{N,k} := \{X\in \mathcal{V}_N: D(X) < D_{N,min}+k\}
\end{equation}
and
\begin{equation} 
\bar{\mathcal{V}}_{N,k} := \mathcal{V}_N \setminus \mathcal{V}_{N,k} = \{ X\in \mathcal{V}_N: D(X) \ge D_{N,min} +k\}.
\end{equation}
By $P_{N,k}$ and $\bar{P}_{N,k} = I-P_{N,k}$ we denote the orthogonal projections onto $\ell^2(\mathcal{V}_{N,k})$ and $\ell^2(\bar{\mathcal{V}}_{N,k})$, respectively. 

We fix $\lambda \ge k(1-\frac{1}{\Delta})/2$ and consider the operator $H_N + \lambda P_{N,k}$, noting that
\begin{eqnarray} \label{bla1}
H_N + \lambda P_{N,k} & \ge & -\frac{1}{2\Delta} \mathcal{L}_N + \frac{1}{2}\left(1-\frac{1}{\Delta}\right)(D_N +kP_{N,k}) \\& \ge & \frac{1}{2}\left(1-\frac{1}{\Delta}\right)(D_{N,min}+k). \notag
\end{eqnarray}
If we fix 
\begin{equation} \label{bla2}
0< \delta \le \frac{D_{N,min}+k}{2}
\end{equation} 
and an energy 
\begin{equation} \label{bla3}
E\le \left(1-\frac{1}{\Delta}\right)\left(\frac{D_{N,min}+k}{2} -\delta\right),
\end{equation}
then $E$ lies strictly below the spectrum of $H_N+\lambda P_{N,k}$. Our goal is to prove the following Combes-Thomas bound for the resolvent $(H_N + \lambda P_{N,k} -E)^{-1}$: 

\begin{theorem} \label{thm:CTgen}
For any $1\le N < |\mathcal{V}|$, $k\in \N_0$, $\delta$ and $E$ as in \eqref{bla2} and \eqref{bla3}, $\lambda \ge k(1-\frac{1}{\Delta})/2$, and for arbitrary subsets ${\mathcal A}$ and ${\mathcal B}$ of ${\mathcal V}_N$ it holds that
\begin{equation} \label{CTbound}
\| \chi_{\mathcal A} (H_N + \lambda P_{N,k} -E)^{-1} \chi_{\mathcal B}\| \le \frac{4D_{N,min}}{\delta(1-1/\Delta)} \left( \frac{\delta(\Delta-1)D_{N,min}}{(D_{N,min}+k)^2 \sqrt{d}} + 1 \right)^{-d_N({\mathcal A}, {\mathcal B})}.
\end{equation}
\end{theorem}

\begin{proof}
As $N$ will be fixed throughout this argument we temporarily drop it from the notation, writing $H=H_N$, $D=D_N$, $D_{min}=D_{N,min}$, $P_k = P_{N,k}$ and $R_{E,k}:= (H+\lambda P_k-E)^{-1}$. 

From the first line of \eqref{bla1} and \eqref{bla3} we get
\begin{eqnarray}
H+\lambda P_k -E & \ge & \frac{1}{2}\left(1-\frac{1}{\Delta}\right) (D+kP_k - D_{min} - k +2\delta) \\
& = & \frac{1}{2}\left(1-\frac{1}{\Delta}\right) (D_k -D_{min}+2\delta), \notag
\end{eqnarray}
where we abbreviate $D_k := D-k\bar{P}_k$. On the other hand, by \eqref{bla2},
\begin{eqnarray}
\frac{\delta D_k}{D_{min}+k} & = & \frac{\delta}{D_{min}+k} (D_k -D_{min}) + \frac{\delta D_{min}}{D_{min}+k} \\
& \le & \frac{1}{2}(D_k -D_{min}+2\delta). \notag
\end{eqnarray}
Combining the last two bounds and noting that $D_k \ge D_{min} \ge 1$, we get
\begin{equation}
\frac{\delta(1-\frac{1}{\Delta})}{D_{min}+k} \le D_k^{-1/2} (H+\lambda P_k -E) D_k^{-1/2}
\end{equation}
or
\begin{equation} \label{doubleplus}
D_k^{1/2} R_{E,k} D_k^{1/2} \le \frac{D_{min}+k}{\delta(1-\frac{1}{\Delta})}.
\end{equation}

For $X\in {\mathcal V}_N$ and ${\mathcal A} \subset {\mathcal V}_N$ let $\rho_{\mathcal A}(X) = d_N({\mathcal A},X)$ and interpret $\rho_{\mathcal A}$ as a multiplication operator on $\ell^2({\mathcal V}_N)$. For $\eta>0$ define the dilation
\begin{equation} 
H_{\eta} := e^{-\eta \rho_{\mathcal A}} H e^{\eta \rho_{\mathcal A}} 
\end{equation}
of $H$ and let $K_{\eta}:= H_{\eta}-H$.  We write $X\sim Y$ to denote next neighbors in ${\mathcal G}_N$ (i.e., $X, Y \in {\mathcal V}_N$, $X\Delta Y \in {\mathcal E}$). A calculation shows that
\begin{equation}
(K_{\eta} \psi)(X) = \frac{1}{2\Delta} \sum_{Y\in {\mathcal V}_N, Y\sim X} \left(1-e^{\eta (\rho_{\mathcal A}(Y) - \rho_{\mathcal A}(X))}\right) \psi(Y)
\end{equation}
for all $\psi \in \ell^2({\mathcal V}_N)$ and $X\in {\mathcal V}_N$. For $X \sim Y$ we have by \eqref{triangle} that $|\rho_{\mathcal A}(X)-\rho_{\mathcal A}(Y)| \le 1$. Thus $|f_{\eta}(X,Y)| \le e^{\eta}-1$, where $f_{\eta}(X,Y) := 1-e^{\eta (\rho_{\mathcal A}(Y) - \rho_{\mathcal A}(X))}$.

We continue by following the calculation (4.16) in \cite{EKS1}:
\begin{eqnarray}
\lefteqn{\langle D^{-1/2} K_\eta D^{-1/2} \psi, D^{-1/2} K_{\eta} D^{-1/2} \psi \rangle} \\
& = & \frac{1}{4\Delta^2} \sum_{X,Y,W\in {\mathcal V}_N, X\sim Y, X\sim W} D^{-1}(X) D^{-1/2}(Y) \times  \notag \\ & &  D^{-1/2}(W) f_{\eta}(X,W) f_{\eta}(X,Y) \psi(Y) \overline{\psi(W)} \notag \\
& \le & \frac{1}{4\Delta^2} (e^{\eta}-1)^2 \sum_{X,Y,W \in {\mathcal V}_N, X\sim Y, X\sim W} D^{-1}(X) D^{-1}(Y) |\psi(Y)|^2, \notag
\end{eqnarray}
which has used that 
\begin{equation} 
|D^{-1/2}(Y)D^{-1/2}(W) \psi(Y) \overline{\psi(W)}| \le \frac{1}{2} D^{-1}(Y) |\psi(Y)|^2 + \frac{1}{2}D^{-1}(W)|\psi(W)|^2.
\end{equation} 
Now using $\sum_{W\in {\mathcal V}_N:\sim X} 1 = D(X)$ (another way of writing the definition of $D(X)$), we see that the latter is equal to
\begin{eqnarray}
\lefteqn{\frac{1}{4\Delta^2} (e^{\eta}-1)^2 \sum_{X,Y\in {\mathcal V}_N, Y\sim X} D^{-1}(Y) |\psi(Y)|^2} \\ 
& \le & \frac{d}{2\Delta^2} (e^{\eta}-1)^2 \sum_{X\in {\mathcal V}_N} \left( \sum_{Y\in {\mathcal V}_N, X\sim Y} D^{-1}(X) \right) |\psi(Y)|^2 \notag \\
& = & \frac{d}{2\Delta^2} (e^{\eta}-1)^2 \|\psi\|^2 \notag
\end{eqnarray}
for all $\psi \in \ell^2({\mathcal V}_N)$, where we have used $D(Y)^{-1} \le 2d D(X)^{-1}$ if $Y\sim X$ (which is \eqref{XYcomp2} expressed in terms of \eqref{degsur}) in the second-to-last step. Thus we arrive at
\begin{equation} \label{doubleminus0}
\|D^{-1/2} K_{\eta} D^{-1/2}\| \le \frac{e^{\eta}-1}{\Delta} \sqrt{d/2}.
\end{equation}
Also noting that $D/D_k$ is a multiplication operator with $\|D/D_k\| \le 1+k/D_{min}$, this further yields
\begin{equation} \label{doubleminus}
\|D_k^{-1/2} K_{\eta} D_k^{-1/2}\| \le \frac{e^{\eta}-1}{\Delta} \sqrt{d/2} \,\Big(1+\frac{k}{D_{min}} \Big).
\end{equation}

Combining \eqref{doubleplus} and \eqref{doubleminus} we arrive at the bound
\begin{eqnarray}
\|D_k^{-1/2} K_{\eta} R_{E,k}D_k^{1/2}\| & = & \|D_k^{-1/2} K_{\eta} D_k^{-1/2} D_k^{1/2}R_{E,k} D_k^{1/2}\|  \\
& \le & \frac{(e^{\eta}-1) \sqrt{d/2} \,(D_{min}+k)^2}{\delta (\Delta-1) D_{min}} = \frac{1}{\sqrt{2}}, \notag
\end{eqnarray}
where we have now made the explicit choice
\begin{equation}
\eta := \log \left( 1+\frac{\delta(\Delta-1) D_{min}}{\sqrt{d} (D_{min}+k)^2} \right).
\end{equation}

Let $R_{\eta,E,k} := (H_{\eta}+\lambda P_k-E)^{-1} = e^{-\eta \rho_{\mathcal A}} R_{E,k} e^{\eta \rho_{\mathcal A}}$. Using that $\|A\|<1$ implies $\|(I+A)^{-1}\| \le 1/(1-\|A\|)$ we obtain
\begin{eqnarray}
\|D_k^{1/2} R_{\eta,E,k} D_k^{1/2}\| & = & \|D_k^{1/2} R_{E,k} D_k^{1/2} (I+D_k^{-1/2} K_{\eta} R_{E,k} D_k^{1/2})^{-1}\| \\
& \le & \frac{1}{1-1/\sqrt{2}} \|D_k^{1/2} R_{E,k} D_k^{1/2}\| \notag \\
& \le & \frac{4D_{N,min}}{\delta(1-1/\Delta)} \notag
\end{eqnarray}
from \eqref{doubleplus}. This allows to conclude the proof via
\begin{eqnarray}
\|\chi_{\mathcal A} R_{E,k} \chi_{\mathcal B} \| & \le & \|\chi_{\mathcal A} D_k^{1/2} R_{E,k} D_k^{1/2} \chi_{\mathcal B}\| \\
& = & \| \chi_{\mathcal A} e^{\eta \rho_{\mathcal A}} D_k^{1/2} R_{\eta,E,k} D_k^{1/2} e^{-\eta \rho_{\mathcal A}} \chi_{\mathcal B}\| \notag \\
& \le & \frac{4D_{min}}{\delta(1-1/\Delta)} e^{-\eta d_N({\mathcal A}, {\mathcal B})} \notag \\
& = & \frac{4D_{min}}{\delta(1-1/\Delta)} \left( \frac{\delta(\Delta-1)D_{min}}{\sqrt{d}(D_{min}+k)^2} + 1 \right)^{-d_N({\mathcal A}, {\mathcal B})}. \notag
\end{eqnarray}

\end{proof}

We include several remarks on modifications/extensions of the above result:

\begin{remark} \label{rem:proj}
{\rm In Theorem~\ref{thm:CTgen} we have regularized the resolvent at energies below the $k$-threshold as in \eqref{bla3} by adding $\lambda P_{N,k}$ to the Hamiltonian. A similar effect can be obtained by projecting the full Hamiltonian onto the complementary space $\ell^2({\bar{\mathcal V}}_{N,k})$. Thus, without going into detail, we mention that a bound similar to \eqref{CTbound} can be found for $\| \chi_{\mathcal A} (\bar{P}_{N,k}(H_N -E)\bar{P}_{N,k})^{-1} \chi_{\mathcal B}\|$ (see \cite{EKS1} for the case of the chain).}
\end{remark}

\begin{remark} \label{rem:field}
{\rm Theorem~\ref{thm:CTgen} generalizes to the case of an XXZ system in an additional exterior field as in Remark~\ref{rem:extfield}, as long as we assume non-negativity $v(x)\ge 0$ for all $x\in {\mathcal V}$. This means that all $N$-body potentials defined via \eqref{eq:potentialdef} are non-negative, so that the proof of the Combes-Thomas bound goes through without any effect on constants. Thus \eqref{CTbound} holds uniformly with respect to the addition of arbitrary positive exterior fields.

A particular situation where this is useful is in the presence of {\it boundary fields}, as encountered when studying an XXZ system on an infinite graph as the ``thermodynamic limit'' of its restrictions to finite subgraphs (as done, e.g., in \cite{EKS1} and \cite{BW1}). Choosing these boundary fields positive will thus allow for Combes-Thomas bounds which are volume independent and stable under the thermodynamic limit.} 
\end{remark}

\begin{remark} \label{rem:complex}
{\rm The bound \eqref{bla3} easily extends to complex energy, i.e., to $\| \chi_{\mathcal A} (H_N + \lambda P_{N,k} -E-i\epsilon)^{-1} \chi_{\mathcal B}\|$ for arbitrary $\epsilon \in \R$, with only an extra factor $2$ in the universal constant on the right hand side. This is seen in the same way as for the case ${\mathcal G} = \Z$ in Section~4 of \cite{EKS1}. This extension also applies in the situations considered in Remarks~\ref{rem:proj} and \ref{rem:field} above.}
\end{remark}

\subsection{Decay bounds on eigenstates and spectral projections} \label{efdecay}

Finally, we show how the above Combes-Thomas bounds imply decay of eigenstates and spectral projections associated with energies below the energies
\begin{equation}
E_{N,k} := \frac{1}{2}\left(1-\frac{1}{\Delta} \right) (D_{N,min}+k), \quad k=1,2,\ldots,
\end{equation}
 which appeared as lower bounds in \eqref{bla1} and play the role of a sequence of $k$-{\it thresholds} for the energies of states characterized by multiple connected components of down-spins. 
 
 \subsubsection{Eigenstates} \label{sub:ef}
 If 
\begin{equation}
E < E_{N,k}
\end{equation}
is an eigenvalue of $H_N$ strictly below one of these thresholds, i.e., $(H_N-E)\psi=0$ for an eigenfunction $\psi \in \ell^2({\mathcal V}_N)$, then we can use Theorem~\ref{thm:CTgen} to show that $\psi$ is concentrated on ${\mathcal V}_{N,k}$, up to exponentially decaying `tails'.  For this choose $\delta = \frac{\Delta}{\Delta-1}(E_{N,k}-E)$, so that \eqref{bla3} holds as an equality. Also, let $\lambda := k(1-\frac{1}{\Delta})/2$ be the smallest constant such that \eqref{bla1} holds. Then, for any ${\mathcal A} \subset {\mathcal V}_N$, Theorem~\ref{thm:CTgen} yields
\begin{eqnarray} \label{eq:esdecay}
\|\chi_{\mathcal A} \psi \| & = & \| \chi_{\mathcal A} (H_N+\lambda P_{N,k}-E)^{-1} (H_N+\lambda P_{N,k}-E) \psi\| \\
& = & |\lambda| \| \chi_{\mathcal A} (H_N+\lambda P_{N,k}-E)^{-1} \chi_{{\mathcal V}_{N,k}}\| \|\chi_{{\mathcal V}_{N,k}} \psi\| \notag \\
& \le & \frac{2(\Delta-1)kD_{N,min}}{\Delta (E_{N,k}-E)} e^{-\gamma d_N({\mathcal A}, {\mathcal V}_{N,k})} \|\chi_{{\mathcal V}_{N,k}} \psi\| \notag
\end{eqnarray}
with 
\begin{equation}
\gamma = \log \left( 1 + \frac{(\Delta-1)(E_{N,k}-E) D_{N,min}}{\sqrt{d} (D_{N,min}+k)^2} \right) .
\end{equation}

Let us discuss some of the properties and special cases of this bound:

(i) For fixed $N$, $k$ and $E<E_{N,k}$, the dependence of \eqref{eq:esdecay} on $\Delta \in (1,\infty)$ takes the form
\begin{equation}
\|\chi_{\mathcal A} \psi \| \le C_1 e^{-\gamma_1 (\Delta-1) d_N({\mathcal A}, {\mathcal V}_{N,k})} \|\chi_{{\mathcal V}_{N,k}} \psi\|.
\end{equation}
with suitable positive constants $C_1$ and $\gamma_1$.

(ii) For fixed $\Delta>1$, $N$ and $k$, the asymptotic behavior for $E \nearrow E_{N,k}$ is given by
\begin{equation}
\|\chi_{\mathcal A} \psi \| \lesssim \frac{C_2}{E_{N,k}-E} e^{-\gamma_2 (E_{N,k}-E) d_N({\mathcal A}, {\mathcal V}_{N,k})} \|\chi_{{\mathcal V}_{N,k}} \psi\|
\end{equation}
with positive constants $C_2$ and $\gamma_2$.

(iii) For $k=1$ we have that $\chi_{{\mathcal V}_{N,1}} \psi$ is the part of $\psi$ which is given as a superposition of classical (strict) droplets, i.e., spin product states where all down-spins form only one connected cluster. The bound \eqref{eq:esdecay} then expresses the exponential closeness of eigenstates to energies in the droplet spectrum to linear combinations of strict droplets, with $d_N$ providing the distance measure.

(iv) For the concrete example of a strip of width $M$ and particle number $N=\ell M$, $\ell=1,2,\ldots$, worked out in Section~\ref{sec:strip}, the minimal surface areas $D_{N,min}$ and thus the threshold energies $E_{N,k}$ do not depend on $N$, so that the Combes-Thomas bounds are uniform in $N$. While we didn't work out any more subtle examples, it becomes clear that $N$-independence of the bounds is essentially a consequence of (quasi) one-dimensionality of the graph $\mathcal{G}$. In fact, the latter could be defined as $\sup_N D_{N,min}<\infty$, the property that surface areas of ``balls'' don't grow with the volume. For graphs $\mathcal{G}$ with this property bounds such as \eqref{CTbound} and \eqref{eq:esdecay} establish the existence of a droplet regime, with states corresponding to arbitrarily large droplets, at the bottom of the spectrum of the XXZ spin system on $\mathcal{G}$. 

Thus, while not working out any other detailed examples, we conclude that our initial question for the existence of a (many-body) droplet regime near the bottom of the spectrum of an XXZ system is answered by the need of quasi-one-dimensionality of the graph. The example $\mathcal{G}=\Z^2$ worked out in Section~\ref{sec:Z2} illustrates that for higher-dimensional XXZ systems the low energy spectrum corresponds to droplets of small particle number and thus is not a many-body energy regime for the model.

\subsubsection{Spectral projections} \label{sub:sp} 
To also get a decay bound on spectral projections of the spectrum of $H_N$ below the thresholds $E_{N,k}$ we require that $H_N$ has a spectral gap below $E_{N,k}$, as it holds for example in the special cases discussed in Section~\ref{sec:examples} if $\Delta$ is sufficiently large. More precisely, assume that 
\begin{equation} \label{eq:dropgap} \left(\frac{1}{2}\Big(1-\frac{1}{\Delta}\Big)(D_{N,min}+k-\delta), E_{N,k}\right) 
\end{equation}
lies in a spectral gap of $H_N$ for some $\delta>0$. Let $Q$ denote the spectral projection for $H_N$ onto $(0,E_{N,k})$. We will prove a bound on $\| \chi_{{\mathcal A}} Q\|$ for any ${\mathcal A} \subset \bar{\mathcal V}_{N,k}$, which decays in the distance of $\mathcal A$ from ${\mathcal V}_{N,k}$. This follows by adapting an argument which was provided in a similar context in \cite{BW2}.

The existence of the spectral gap \eqref{eq:dropgap} allows to choose a positively oriented integration contour $\Gamma$ around $\sigma(H_N) \cap (0,E_{N,k})$ which cuts the positive real line at the centers of the ground state gap $(0, (1-1/\Delta)/2)$ and the gap \eqref{eq:dropgap}. This means $\Gamma$ can be chosen such that its length is bounded by $\ell(\Gamma) \le D_{N,min} + k +2$ and so that $\|(H_N-z)^{-1}\| \le 4/((1-1/\Delta) \min \{1,\delta\})$ for all $z\in \Gamma$. Thus 
\begin{eqnarray} \label{eq:intsplit}
\chi_{\mathcal A} Q & = & \frac{1}{2\pi i} \oint_{\Gamma} \chi_{\mathcal A} (z-H_N)^{-1}\,dz  \\
& = & \frac{1}{2\pi i} \oint_{\Gamma} \chi_{\mathcal A} (z-H_N)^{-1} \bar{P}_{N,k}\,dz + \frac{1}{2\pi i} \oint_{\Gamma} \chi_{\mathcal A} (z-H_N)^{-1} P_{N,k}\,dz \notag
\end{eqnarray}
By arguing with Schur complementation over $\ell^2({\mathcal V}_{N,k}) \oplus \ell^2(\bar{\mathcal V}_{N,k})$ (i.e., \eqref{eq:SchurInv} at complex energy) one sees that $ \chi_{\mathcal A} (z-H_N)^{-1} \bar{P}_{N,k}$ is analytic in the interior of $\Gamma$, so that the first term in \eqref{eq:intsplit} vanishes. The second term in \eqref{eq:intsplit} decomposes further into
\begin{equation}
\frac{1}{2\pi i} \oint_{\Gamma} \chi_{\mathcal A} (H_N+\lambda P_{N,k}-z)^{-1} P_{N,k}\, dz -\frac{\lambda}{2\pi i} \oint_{\Gamma} \chi_{\mathcal A} (H_N+\lambda P_{N,k} -z)^{-1} P_{N,k} (H_N-z)^{-1}\,dz.
\end{equation}
Thus we arrive at the norm bound
\begin{eqnarray} 
\|\chi_{\mathcal A} Q \| & \le & C_0\left( 1 + \frac{|\lambda|}{(1-1/\Delta) \rm{min}(1,\delta)}\right) \ell(\Gamma)  \|\chi_{\mathcal A} (H_N+\lambda P_{N,k} -z)^{-1} \chi_{{\mathcal V}_{N,k}}\| \\
& \le & \frac{C(k+2)(D_{N,min}+k+2)}{\min \{1,\delta\}} \|\chi_{\mathcal A} (H_N+\lambda P_{N,k} -z)^{-1} \chi_{{\mathcal V}_{N,k}}\|, \notag
\end{eqnarray}
where we have chosen the smallest allowed value $k(1-1/\Delta)/2$ for $\lambda$ and $C_0$ and $C$ are universal constants. We can now use \eqref{CTbound} (and its extension to complex energy in Remark~\ref{rem:complex}) to get that $\|\chi_{\mathcal A} Q\|$ decays exponentially in the distance $d_N({\mathcal A}, {\mathcal V}_{N,k})$, similar to \eqref{eq:esdecay}.

The remarks on eigenstates at the end of Section~\ref{sub:ef}, relating a many-body droplet regime at the bottom of the spectrum to quasi-one-dimensionality of the graph, also apply in the context of spectral projections.

\section*{Acknowledgement}

We are grateful to Luis Manuel Rivera for informing us about several other contexts in which symmetric product graphs (under other names) have appeared in the graph theoretic literature and for providing us with relevant references.

\begin{appendix} 

\section{Proof of the distance formula \eqref{distN}} \label{AppA}

To avoid the use of labelings, we restate \eqref{distN} as follows: Let $X, Y \subset \mathcal{V}$ with $|X|=|Y|=N<\infty$ and $S_{X,Y}=\{\pi:X \to Y \;\mbox{bijective}\}$. Then
\begin{equation} \label{distN2}
d_N(X,Y) = \min_{\pi \in S_{X,Y}} \sum_{x\in X} d(x,\pi(x)).
\end{equation}

To prove the lower bound in \eqref{distN2}, let $(X=Z_0,Z_1,\ldots,Z_{k-1},Z_k=Y)$ be a shortest path from $X$ to $Y$ in $\mathcal{G}_N$. By Remark~\ref{neighbor} this means that, for each $\ell = 0,\ldots,k-1$, $Z_{\ell+1}$ is found by moving one of the elements of $Z_\ell$ to a next neighbor. This yields the existence of $\pi' \in S_{X,Y}$ such that for each $x\in X$ there is a path of length $\ell(x)$ from $x$ to $\pi'(x)$ in $\mathcal{G}$ with $d_N(X,Y)=k=\sum_x \ell(x)$. Thus
\begin{equation} 
d_N(X,Y) \ge \sum_x d(x,\pi'(x)) \ge \min_{\pi\in S_{X,Y}} \sum_{x} d(x, \pi(x)).
\end{equation}

For the proof of the upper bound in \eqref{distN2} we first claim that 
\begin{equation} \label{abridge}
\min_{\pi \in S_{X,Y}} \sum_{x\in X} d(x,\pi(x)) = \min_{\sigma \in S_{X,Y}'} \sum_{x \in X\setminus Y} d(x,\sigma(x)),
\end{equation}
where $S_{X,Y}' = \{\sigma \in S_{X,Y}: \sigma|_{X\cap Y} = id\}$. To see this, let $\pi \in S_{X,Y}$. If $\pi \not\in S_{X,Y}'$, then there exist finitely many orbits of $\pi$ of the form $(z_0,z_1,\ldots,z_n,z_{n+1})$ with $n\ge 1$, $z_0 \in X\setminus Y$, $z_1,\ldots, z_n \in X\cap Y$ and $z_{n+1} \in Y\setminus X$, see Figure~\ref{figA1}. 

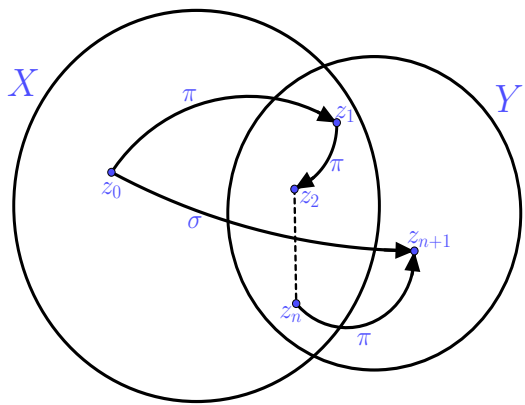
\begin{figure}[h]  
\centering
\resizebox{8cm}{6cm}
{
\begin{tikzpicture}[line cap=round,line join=round,>=triangle 45,x=1.0cm,y=1.0cm] 
\clip(9.,2.) rectangle (23.,12.);
\draw [line width=2.pt] (18.235095306300096,7.090892803147953) circle (3.4748814113937687cm);
\draw [line width=2.pt] (14.019466511511185,7.253519058059407) circle (4.336707331605724cm);
\draw [shift={(15.26366255995327,5.679828498440437)},line width=2pt]  plot[domain=1.0242448534356268:2.5235817744399727,variable=\t]({1.*4.004333790019244*cos(\t r)+0.*4.004333790019244*sin(\t r)},{0.*4.004333790019244*cos(\t r)+1.*4.004333790019244*sin(\t r)});
\draw [->,line width=2pt] (16.895824150671544,9.336430424228077) -- (17.344891715058328,9.100818750667209);
\draw [shift={(15.838054610522198,9.048173987086376)},line width=2pt]  plot[domain=-1.2300359639970875:0.034923058026760836,variable=\t]({1.*1.5077564560429877*cos(\t r)+0.*1.5077564560429877*sin(\t r)},{0.*1.5077564560429877*cos(\t r)+1.*1.5077564560429877*sin(\t r)});
\draw [dashed, ultra thick] (16.34195259471438,7.627112288121)-- (16.382888885340666,5.08906226929142);
\draw [shift={(17.581121139300624,6.16232530367498)},line width=2pt]  plot[domain=-2.4111553626479267:0.05810814191497408,variable=\t]({1.*1.608618685520024*cos(\t r)+0.*1.608618685520024*sin(\t r)},{0.*1.608618685520024*cos(\t r)+1.*1.608618685520024*sin(\t r)});
\draw [shift={(19.35836630408279,22.640583547084333)},line width=2pt]  plot[domain=4.246663043840436:4.701932039886118,variable=\t]({1.*16.38573285709911*cos(\t r)+0.*16.38573285709911*sin(\t r)},{0.*16.38573285709911*cos(\t r)+1.*16.38573285709911*sin(\t r)});
\draw [->,line width=2pt] (19.129524428300595,5.726319191311589) -- (19.187024793241083,6.255746552140502);
\draw [->,line width=2pt] (18.478285185374823,6.278502429796376) -- (19.187024793241083,6.255746552140502);
\draw [->,line width=2pt] (16.79348688949084,7.881778389081951) -- (16.34195259471438,7.627112288121);
\begin{scriptsize}
\draw [fill=ududff] (12.,8.) circle (2.5pt);
\draw[color=ududff] (12.0,7.60945301468671) node {\Large $z_0$};
\draw [fill=ududff] (17.344891715058328,9.100818750667209) circle (2.5pt);
\draw[color=ududff] (17.539339095533172,9.3) node {\Large $z_1$};
\draw [fill=ududff] (16.34195259471438,7.627112288121) circle (2.5pt);
\draw[color=ududff] (16.700145137694363,7.4) node {\Large $z_2$};
\draw [fill=ududff] (16.382888885340666,5.08906226929142) circle (2.5pt);
\draw[color=ududff] (16.256868120502372,4.845785252099427) node {\Large $z_n$};
\draw [fill=ududff] (19.187024793241083,6.255746552140502) circle (2.5pt);
\draw[color=ududff] (19.54521733622106,6.5) node {\Large $z_{n+1}$};
\draw[color=ududff] (9.914954966387876,10) node {\Huge $X$};
\draw[color=ududff] (21.441459898030622,9.643224601465466) node {\Huge $Y$};
\draw[color=ududff] (17.35507293916766,8.16098597576865) node {\Large $\pi$};
\draw[color=ududff] (13.85507293916766,9.706098597576865) node {\Large $\pi$};
\draw[color=ududff] (13.988115883703083,6.900493129504468) node {\Large $\sigma$};
\draw[color=ududff] (18,4.3) node {\Large $\pi$};
\end{scriptsize}
\end{tikzpicture}
}
\caption{Abridging $X\cap Y$}
\label{figA1}
\end{figure}

For each of these orbits redefine $\pi$ on $z_0, \ldots, z_n$ as $\sigma(z_0) = z_{n+1}$, $\sigma(z_j)=z_j$, $j=1,\ldots,n$. This yields $\sigma \in S_{X,Y}'$ such that, by the triangle inequality in $\mathcal{G}$,
\begin{equation} \label{abridgege}
 \sum_{x\in X} d(x,\pi(x)) \ge \sum_{x\in X\setminus Y} d(x,\sigma(x)).
 \end{equation}
This proves the lower bound in \eqref{abridge}, while the upper bound is trivial.

Now we claim that for every $\sigma \in S_{X,Y}'$ there exists a path from $X$ to $Y$ in $\mathcal{G}_N$ of length $\sum_{x\in X\setminus Y} d(x,\sigma(x))$, which together with \eqref{abridge} completes the proof of the upper bound in \eqref{distN2}. 

Let $X\setminus Y = \{x_1,\ldots,x_{\ell}\}$ and $Y\setminus X = \{y_1=\sigma(x_1), \ldots, y_{\ell} = \sigma(x_{\ell})\}$. Let $p=(x_1=z_0,z_1,\ldots,z_n=y_1)$ be a shortest path from $x_1$ to $y_1$ in $\mathcal{G}$, and $\{z_{j_1}, \ldots, z_{j_r}\} = \{z_1,\ldots,z_{n-1}\} \cap X$, see Figure~\ref{figA2}. 

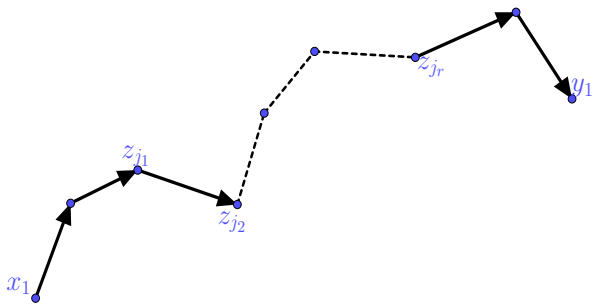
\begin{figure}[h] 
\centering
\resizebox{9cm}{5cm}
{
\begin{tikzpicture}[line cap=round,line join=round,>=triangle 45,x=1.0cm,y=1.0cm]
\clip(1.,0.) rectangle (16.,8.);
\draw [->,line width=2.pt] (2.4,1.2) -- (3.189634189152375,3.2205118799727783);
\draw [->,line width=2.pt] (3.189634189152375,3.220511879972778) -- (4.704866618562552,3.92762034703086);
\draw [->,line width=2.pt] (4.704866618562552,3.92762034703086) -- (6.952461388854316,3.195258006149275);
\draw [->,line width=2.pt] (10.967827326791287,6.32673836026364) -- (13.240675970906553,7.286385565556752);
\draw [->,line width=2.pt] (13.240675970906553,7.286385565556752) -- (14.5033696620817,5.442852776441037);
\draw [dashed, ultra thick] (6.952461388854316,3.195258006149275)-- (7.558554360618387,5.139806290559002);
\draw [dashed, ultra thick] (7.558554360618387,5.139806290559002)-- (8.69497868267602,6.453007729381155);
\draw [dashed, ultra thick] (8.69497868267602,6.453007729381155)-- (10.967827326791287,6.32673836026364);
\begin{scriptsize}
\draw [fill=ududff] (2.4,1.2) circle (2.5pt);
\draw[color=ududff] (2.023453018711582,1.4537562076215832) node {\Large $x_1$};
\draw [fill=ududff] (3.189634189152375,3.220511879972778) circle (2.5pt);
\draw [fill=ududff] (4.704866618562552,3.92762034703086) circle (2.5pt);
\draw[color=ududff] (4.66800130312131,4.235744309382954) node {\Large $z_{j_1}$};
\draw [fill=ududff] (6.952461388854316,3.195258006149275) circle (2.5pt);
\draw[color=ududff] (6.837803461354541,2.82762034703086) node {\Large $z_{j_2}$};
\draw [fill=ududff] (7.558554360618387,5.139806290559002) circle (2.5pt);
\draw [fill=ududff] (8.69497868267602,6.453007729381155) circle (2.5pt);
\draw [fill=ududff] (10.967827326791287,6.32673836026364) circle (2.5pt);
\draw[color=ududff] (11.334008497232078,6.124707369675616) node {\Large $z_{j_r}$};
\draw [fill=ududff] (13.240675970906553,7.286385565556752) circle (2.5pt);
\draw [fill=ududff] (14.5033696620817,5.442852776441037) circle (2.5pt);
\draw[color=ududff] (14.74328146340498,5.673184126734599) node {\Large $y_1$};
\end{scriptsize}
\end{tikzpicture}
}
\caption{Obstacles $z_{j_1}, \ldots, z_{j_r}$ in a shortest path}
\label{figA2}
\end{figure}

Then a path in $\mathcal{G}_N$ from $X$ to $(X\setminus \{x_1\}) \cup \{y_1\}$ of length $n=d(x_1,\sigma(x_1))$ is found by successively moving $z_{j_r}$ to $y_1$, $z_{j_{r-1}}$ to $z_{j_r}$, \ldots, and $x_1$ to $z_{j_1}$ along $p$ (resulting in a ``caterpillar'' move, to use a term from \cite{Johns}). 

In this way each $x_j$ can be moved to $y_j$, for $j=1,\ldots,\ell$, giving a path of the desired length in $\mathcal{G}_N$.

\section{Solutions to certain edge-isoperimetric problems} \label{sec:AppB}

Given a graph $\mathcal{G}=(\mathcal{V,E})$ and its $N$-th symmetric product $\mathcal{G}_N=(\mathcal{V}_N,\mathcal{E}_N)$, recall that finding the lowest eigenvalue of the operator $D_N$ is equivalent to solving the edge-isoperimetric problem
\begin{equation} \label{eq:minimum} 
\min\{|\partial X|:X\subset\mathcal{V}, |X|=N\}=:D_{N,min}\:.
\end{equation}
In addition to this, we will also determine the set of all minimizers $\mathcal{V}_{N,1}$ of \eqref{eq:minimum}.

\subsection{Classical droplets on the two-dimensional lattice}

Let $\mathcal{G}=(\mathcal{V},\mathcal{E})$ be the two-dimensional lattice, i.e.\ the vertex set is given by $\mathcal{V}=\Z^2$ and the edge set is given by $\ell^1$-next-neighbors: $\mathcal{E}=\{\{x,y\}: x,y\in\Z^2, \|x-y\|_1=1\}$. 

For our arguments, it will be convenient to identify vertices $x\in\Z^2$ with the unit square in $\R^2$ with $x$ as center point (cf. Figures \ref{fig:Z2disconnected} and \ref{fig:z2connected}). To make this more precise, let 
\begin{equation} 
\mathcal{U}_x=[x_1-1/2,x_1+1/2]\times[x_2-1/2,x_2+1/2]\subset \R^2
\end{equation} 
denote the unit square with center $x$ and for any $X\subset\Z^2$, we define
\begin{equation} 
\mathcal{U}(X):=\bigcup_{x\in X}\mathcal{U}_x\:.
\end{equation}
For any such $\mathcal{U}(X)$ let $\mbox{Area}(X)$ denote its area and $\mbox{Per}(X)$ its perimeter. Note that $\mbox{Area}(X)=|X|$ and $\mbox{Per}(X)=S(X)=|\partial X|$. 

If $N$ is a square (i.e.\ if $\sqrt{N}\in\N$), let us now prove that the set of all classical droplet configurations is given by the quadratic ones. For any $x=(x_1,x_2)\in\Z^2$, we thus define 
\begin{equation}
Q_{x}:=\{x_1,x_1+1,\dots,x_1+\sqrt{N}-1\}\times\{x_2,x_2+1,\dots,x_2+\sqrt{N}-1\}\:.
\end{equation}
\begin{lemma} Let $\mathcal{G}$ be the two-dimensional lattice. Then, the set of all classical droplet configurations $\mathcal{V}_{N,1}$ is given by
\begin{equation}
\mathcal{V}_{N,1}=\{Q_x: x\in\Z^2\}\:,
\end{equation}
which means in particular that $D_{N,min}=4\sqrt{N}$. 
\label{lemma:square}
\end{lemma}

We include the proof of this result for convenience. While we restrict ourselves to squares, the solution to the isoperimetric problem on $\Z^2$ is known for arbitrary $N$, where the minimizers grow in a spiraling pattern with $N$, see Section 2.1 of \cite{Bezrukov} and further references provided there.

\begin{proof}
For a proof by contradiction, assume that $\widehat{X}$ is a non-quadratic classical droplet configuration, which means that $\widehat{X}\in\mathcal{V}_{N,1}$ but there is no $x\in\Z^2$ such that $\widehat{X}=Q_x$. 
Firstly, note that any $\widehat{X}\in\mathcal{V}_{N,1}$ has to be connected, i.e.\ for any decomposition $\widehat{X}=\widehat{X}_1\cup\widehat{X}_2$, there is always at least one $x_1\in\widehat{X}_1$ and one $x_2\in\widehat{X}_2$ such that $\{x_1,x_2\}\in\mathcal{E}$. For if this is not the case, one could always find a new set $\widehat{X}'=\widehat{X}_1\cup\widehat{X}_2'$ with smaller edge-surface, where $\widehat{X}_2'$ is a suitable translate of $\widehat{X}_2$ which is adjoined to $\widehat{X}_1$ (see Figures~\ref{fig:Z2disconnected} and \ref{fig:z2connected}). 

{\definecolor{ffqqqq}{rgb}{1.,0.,0.}
\definecolor{rechteck}{rgb}{0.,0.,1.}
\definecolor{ududff}{rgb}{0.30196078431372547,0.30196078431372547,1.}

\begin{figure}[h]
\centering
\resizebox{9cm}{5cm}
{
\begin{tikzpicture}[line cap=round,line join=round,>=triangle 45,x=1.0cm,y=1.0cm, scale=1]
\clip(2.,-8.) rectangle (20.,4.);
\fill[line width=2.pt,color=ffqqqq,fill=ffqqqq,fill opacity=0.10000000149011612] (4.,-2.) -- (4.,-4.) -- (6.,-4.) -- (6.,-6.) -- (10.,-6.) -- (10.,-2.) -- (8.,-2.) -- (8.,0.) -- (6.,0.) -- (6.,2.) -- (4.,2.) -- cycle;
\fill[line width=2.pt,color=ffqqqq,fill=ffqqqq,fill opacity=0.10000000149011612] (10.,2.) -- (10.,0.) -- (12.,0.) -- (12.,-2.) -- (12.,-4.) -- (12.,-6.) -- (14.,-6.) -- (14.,0.) -- (16.,0.) -- (16.,-2.) -- (18.,-2.) -- (18.,2.) -- cycle;
\draw [line width=1.pt,domain=2.:20.] plot(\x,{(-140.-0.*\x)/20.});
\draw [line width=1.pt,domain=2.:20.] plot(\x,{(--100.-0.*\x)/-20.});
\draw [line width=1.pt,domain=2.:20.] plot(\x,{(-60.-0.*\x)/20.});
\draw [line width=1.pt,domain=2.:20.] plot(\x,{(--20.-0.*\x)/-20.});
\draw [line width=1.pt,domain=2.:20.] plot(\x,{(--20.-0.*\x)/20.});
\draw [line width=1.pt,domain=2.:20.] plot(\x,{(--60.-0.*\x)/20.});
\draw [line width=1.pt,domain=2.:20.] plot(\x,{(-100.-0.*\x)/-20.});
\draw [line width=1.pt] (3.,-8.) -- (3.,4.);
\draw [line width=1.pt] (5.010346226420869,-8.) -- (5.010346226420869,4.);
\draw [line width=1.pt] (7.,-8.) -- (7.,4.);
\draw [line width=1.pt] (9.,-8.) -- (9.,4.);
\draw [line width=1.pt] (11.,-8.) -- (11.,4.);
\draw [line width=1.pt] (13.,-8.) -- (13.,4.);
\draw [line width=1.pt] (15.,-8.) -- (15.,4.);
\draw [line width=1.pt] (17.,-8.) -- (17.,4.);
\draw [line width=1.pt] (19.,-8.) -- (19.,4.);
\draw [line width=2.pt,color=ffqqqq] (4.,-2.)-- (4.,-4.);
\draw [line width=2.pt,color=ffqqqq] (4.,-4.)-- (6.,-4.);
\draw [line width=2.pt,color=ffqqqq] (6.,-4.)-- (6.,-6.);
\draw [line width=2.pt,color=ffqqqq] (6.,-6.)-- (10.,-6.);
\draw [line width=2.pt,color=ffqqqq] (10.,-6.)-- (10.,-2.);
\draw [line width=2.pt,color=ffqqqq] (10.,-2.)-- (8.,-2.);
\draw [line width=2.pt,color=ffqqqq] (8.,-2.)-- (8.,0.);
\draw [line width=2.pt,color=ffqqqq] (8.,0.)-- (6.,0.);
\draw [line width=2.pt,color=ffqqqq] (6.,0.)-- (6.,2.);
\draw [line width=2.pt,color=ffqqqq] (6.,2.)-- (4.,2.);
\draw [line width=2.pt,color=ffqqqq] (4.,2.)-- (4.,-2.);
\draw [line width=1.pt,dotted,color=ffqqqq] (6.,0.)-- (4.,0.);
\draw [line width=1.pt,dotted,color=ffqqqq] (4.,-2.)-- (8.,-2.);
\draw [line width=1.pt,dotted,color=ffqqqq] (6.,0.)-- (6.,-4.);
\draw [line width=1.pt,dotted,color=ffqqqq] (8.,-2.)-- (8.,-6.);
\draw [line width=1.pt,dotted,color=ffqqqq] (6.,-4.)-- (10.,-4.);
\draw (2.884363351309636,3.08967050360574424) node[anchor=north west] {\LARGE $\widehat{X}_1$};
\draw (17.72160482015922,3.08967050360574424) node[anchor=north west] {\LARGE $\widehat{X}_2$};
\draw [line width=2.pt,color=ffqqqq] (10.,2.)-- (10.,0.);
\draw [line width=2.pt,color=ffqqqq] (10.,0.)-- (12.,0.);
\draw [line width=2.pt,color=ffqqqq] (12.,0.)-- (12.,-2.);
\draw [line width=2.pt,color=ffqqqq] (12.,-2.)-- (12.,-4.);
\draw [line width=2.pt,color=ffqqqq] (12.,-4.)-- (12.,-6.);
\draw [line width=2.pt,color=ffqqqq] (12.,-6.)-- (14.,-6.);
\draw [line width=2.pt,color=ffqqqq] (14.,-6.)-- (14.,0.);
\draw [line width=2.pt,color=ffqqqq] (14.,0.)-- (16.,0.);
\draw [line width=2.pt,color=ffqqqq] (16.,0.)-- (16.,-2.);
\draw [line width=2.pt,color=ffqqqq] (16.,-2.)-- (18.,-2.);
\draw [line width=2.pt,color=ffqqqq] (18.,-2.)-- (18.,2.);
\draw [line width=2.pt,color=ffqqqq] (18.,2.)-- (10.,2.);
\draw [line width=1.pt,dotted,color=ffqqqq] (12.,2.)-- (12.,0.);
\draw [line width=1.pt,dotted,color=ffqqqq] (14.,0.)-- (14.,2.);
\draw [line width=1.pt,dotted,color=ffqqqq] (16.,0.)-- (16.,2.);
\draw [line width=1.pt,dotted,color=ffqqqq] (16.,0.)-- (18.,0.);
\draw [line width=1.pt,dotted,color=ffqqqq] (14.,0.)-- (12.,0.);
\draw [line width=1.pt,dotted,color=ffqqqq] (12.,-2.)-- (14.,-2.);
\draw [line width=1.pt,dotted,color=ffqqqq] (14.,-4.)-- (12.,-4.);
\begin{scriptsize}
\draw[color=ffqqqq] (5,1) node {\Large$\bullet$};
\draw[color=ffqqqq] (5,-1) node {\Large$\bullet$};
\draw[color=ffqqqq] (7,-1) node {\Large$\bullet$};
\draw[color=ffqqqq] (5,-3) node {\Large$\bullet$};
\draw[color=ffqqqq] (7,-3) node {\Large$\bullet$};
\draw[color=ffqqqq] (9,-3) node {\Large$\bullet$};
\draw[color=ffqqqq] (7,-5) node {\Large$\bullet$};
\draw[color=ffqqqq] (9,-5) node {\Large$\bullet$};
\draw[color=ffqqqq] (11,1) node {\Large$\bullet$};
\draw[color=ffqqqq] (13,1) node {\Large$\bullet$};
\draw[color=ffqqqq] (15,1) node {\Large$\bullet$};
\draw[color=ffqqqq] (17,1) node {\Large$\bullet$};
\draw[color=ffqqqq] (13,-1) node {\Large$\bullet$};
\draw[color=ffqqqq] (13,-3) node {\Large$\bullet$};
\draw[color=ffqqqq] (13,-5) node {\Large$\bullet$};
\draw[color=ffqqqq] (17,-1) node {\Large$\bullet$};
\draw [fill=ududff] (1.,3.) circle (2.5pt);
\draw[color=ududff] (1.2227464215778143,3.596279928461965) node {$A$};
\draw [fill=ududff] (1.,1.) circle (2.5pt);
\draw[color=ududff] (1.2227464215778143,1.5981720374790913) node {$B$};
\draw [fill=ududff] (1.,-1.) circle (2.5pt);
\draw[color=ududff] (1.2227464215778143,-0.39993585350378247) node {$C$};
\draw [fill=ududff] (1.,-3.) circle (2.5pt);
\draw[color=ududff] (1.2227464215778143,-2.398043744486656) node {$D$};
\draw [fill=ududff] (1.,-5.) circle (2.5pt);
\draw[color=ududff] (1.2227464215778143,-4.39615163546953) node {$E$};
\draw [fill=ududff] (1.,-7.) circle (2.5pt);
\draw[color=ududff] (1.2227464215778143,-6.394259526452404) node {$F$};
\draw [fill=ududff] (1.,5.) circle (2.5pt);
\draw[color=ududff] (1.2227464215778143,5.594387819444838) node {$G$};
\draw [fill=ududff] (21.,5.) circle (2.5pt);
\draw[color=ududff] (21.236052878035306,5.594387819444838) node {$H$};
\draw [fill=ududff] (21.,3.) circle (2.5pt);
\draw[color=ududff] (21.236052878035306,3.596279928461965) node {$I$};
\draw [fill=ududff] (21.,1.) circle (2.5pt);
\draw[color=ududff] (21.236052878035306,1.5981720374790913) node {$J$};
\draw [fill=ududff] (21.,-1.) circle (2.5pt);
\draw[color=ududff] (21.236052878035306,-0.39993585350378247) node {$K$};
\draw [fill=ududff] (21.,-3.) circle (2.5pt);
\draw[color=ududff] (21.236052878035306,-2.398043744486656) node {$L$};
\draw [fill=ududff] (21.,-5.) circle (2.5pt);
\draw[color=ududff] (21.236052878035306,-4.39615163546953) node {$M$};
\draw [fill=ududff] (21.,-7.) circle (2.5pt);
\draw[color=ududff] (21.236052878035306,-6.394259526452404) node {$N$};
\draw [fill=ududff] (3.,7.) circle (2.5pt);
\draw[color=ududff] (3.220854312560687,7.592495710427713) node {$O$};
\draw [fill=ududff] (5.,7.) circle (2.5pt);
\draw[color=ududff] (5.218962203543561,7.592495710427713) node {$P$};
\draw [fill=ududff] (7.,7.) circle (2.5pt);
\draw[color=ududff] (7.217070094526434,7.592495710427713) node {$Q$};
\draw [fill=ududff] (9.,7.) circle (2.5pt);
\draw[color=ududff] (9.215177985509309,7.592495710427713) node {$R$};
\draw [fill=ududff] (11.,7.) circle (2.5pt);
\draw[color=ududff] (11.213285876492181,7.592495710427713) node {$S$};
\draw [fill=ududff] (13.,7.) circle (2.5pt);
\draw[color=ududff] (13.211393767475055,7.592495710427713) node {$T$};
\draw [fill=ududff] (15.,7.) circle (2.5pt);
\draw[color=ududff] (15.20950165845793,7.592495710427713) node {$U$};
\draw [fill=ududff] (17.,7.) circle (2.5pt);
\draw[color=ududff] (17.239837096069557,7.592495710427713) node {$V$};
\draw [fill=ududff] (19.,7.) circle (2.5pt);
\draw[color=ududff] (19.23794498705243,7.592495710427713) node {$W$};
\draw [fill=ududff] (3.,-9.) circle (2.5pt);
\draw[color=ududff] (3.220854312560687,-8.392367417435278) node {$Z$};
\draw [fill=ududff] (5.023766397005074,-9.079754324721877) circle (2.5pt);
\draw[color=ududff] (5.331758616744207,-8.408481190749656) node {$A_1$};
\draw [fill=ududff] (7.,-9.) circle (2.5pt);
\draw[color=ududff] (7.2976389610983246,-8.311798550863388) node {$B_1$};
\draw [fill=ududff] (9.,-9.) circle (2.5pt);
\draw[color=ududff] (9.295746852081198,-8.311798550863388) node {$C_1$};
\draw [fill=ududff] (11.,-9.) circle (2.5pt);
\draw[color=ududff] (11.29385474306407,-8.311798550863388) node {$D_1$};
\draw [fill=ududff] (13.,-9.) circle (2.5pt);
\draw[color=ududff] (13.291962634046946,-8.311798550863388) node {$E_1$};
\draw [fill=ududff] (15.,-9.) circle (2.5pt);
\draw[color=ududff] (15.290070525029819,-8.311798550863388) node {$F_1$};
\draw [fill=ududff] (17.,-9.) circle (2.5pt);
\draw[color=ududff] (17.32040596264145,-8.311798550863388) node {$G_1$};
\draw [fill=ududff] (19.,-9.) circle (2.5pt);
\draw[color=ududff] (19.31851385362432,-8.311798550863388) node {$H_1$};
\draw[color=black] (-4.674894611484538,-6.458714619709916) node {$f$};
\draw[color=black] (-4.674894611484538,-4.460606728727043) node {$g$};
\draw[color=black] (-4.674894611484538,-2.4624988377441683) node {$h$};
\draw[color=black] (-4.674894611484538,-0.46439094676129444) node {$i$};
\draw[color=black] (-4.674894611484538,1.5337169442215792) node {$j$};
\draw[color=black] (-4.674894611484538,2.7583637161143084) node {$k$};
\draw[color=black] (-4.674894611484538,4.788699153725938) node {$l$};
\draw[color=black] (3.3497644990757114,12.97449599742997) node {$m$};
\draw[color=black] (5.315644843429829,12.97449599742997) node {$n$};
\draw[color=black] (7.3459802810414585,12.97449599742997) node {$p$};
\draw[color=black] (9.344088172024332,12.97449599742997) node {$q$};
\draw[color=black] (11.342196063007206,12.97449599742997) node {$r$};
\draw[color=black] (13.340303953990079,12.97449599742997) node {$s$};
\draw[color=black] (15.338411844972953,12.97449599742997) node {$t$};
\draw[color=black] (17.336519735955825,12.97449599742997) node {$a$};
\draw[color=black] (18.561166507848554,12.97449599742997) node {$b$};
\end{scriptsize}
\end{tikzpicture}
}
\caption{An example of a set $\widehat{X}=\widehat{X}_1\cup\widehat{X}_2$ which is not connected.}
\label{fig:Z2disconnected}
\end{figure}
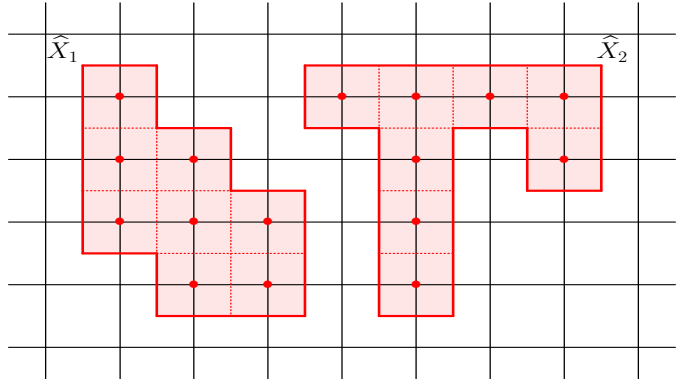

\begin{figure}[h]
\centering
\resizebox{9cm}{5cm}
{
\begin{tikzpicture}[line cap=round,line join=round,>=triangle 45,x=1.0cm,y=1.0cm, scale=1]
\clip(2.,-8.) rectangle (20.,4.);
\fill[line width=2.pt,color=ffqqqq,fill=ffqqqq,fill opacity=0.10000000149011612] (4.,-2.) -- (4.,-4.) -- (6.,-4.) -- (6.,-6.) -- (10.,-6.) -- (10.,-2.) -- (8.,-2.) -- (8.,0.) -- (6.,0.) -- (6.,2.) -- (4.,2.) -- cycle;
\fill[line width=2.pt,color=ffqqqq,fill=ffqqqq,fill opacity=0.10000000149011612] (10.,-6.) -- (12.,-6.) -- (12.,0.) -- (14.,0.) -- (14.,-2.) -- (16.,-2.) -- (16.,2.) -- (10.,2.) -- (8.,2.) -- (8.,0.) -- (10.,0.) -- (10.,-2.) -- cycle;
\draw [line width=1.pt,domain=2.:20.] plot(\x,{(-140.-0.*\x)/20.});
\draw [line width=1.pt,domain=2.:20.] plot(\x,{(--100.-0.*\x)/-20.});
\draw [line width=1.pt,domain=2.:20.] plot(\x,{(-60.-0.*\x)/20.});
\draw [line width=1.pt,domain=2.:20.] plot(\x,{(--20.-0.*\x)/-20.});
\draw [line width=1.pt,domain=2.:20.] plot(\x,{(--20.-0.*\x)/20.});
\draw [line width=1.pt,domain=2.:20.] plot(\x,{(--60.-0.*\x)/20.});
\draw [line width=1.pt,domain=2.:20.] plot(\x,{(-100.-0.*\x)/-20.});
\draw [line width=1.pt] (3.,-8.) -- (3.,4.);
\draw [line width=1.pt] (5.010346226420869,-8.) -- (5.010346226420869,4.);
\draw [line width=1.pt] (7.,-8.) -- (7.,4.);
\draw [line width=1.pt] (9.,-8.) -- (9.,4.);
\draw [line width=1.pt] (11.,-8.) -- (11.,4.);
\draw [line width=1.pt] (13.,-8.) -- (13.,4.);
\draw [line width=1.pt] (15.,-8.) -- (15.,4.);
\draw [line width=1.pt] (17.,-8.) -- (17.,4.);
\draw [line width=1.pt] (19.,-8.) -- (19.,4.);
\draw [line width=2.pt,color=ffqqqq] (4.,-2.)-- (4.,-4.);
\draw [line width=2.pt,color=ffqqqq] (4.,-4.)-- (6.,-4.);
\draw [line width=2.pt,color=ffqqqq] (6.,-4.)-- (6.,-6.);
\draw [line width=2.pt,color=ffqqqq] (6.,-6.)-- (10.,-6.);
\draw [line width=1.pt,dotted,color=ffqqqq] (10.,-6.)-- (10.,-2.);
\draw [line width=2.pt,color=rechteck] (16.1,2.1)-- (3.9,2.1);
\draw [line width=2.pt,color=rechteck] (16.1,2.1)-- (16.1,-6.1);
\draw [line width=2.pt,color=rechteck] (3.9,2.1)-- (3.9,-6.1);
\draw [line width=2.pt,color=rechteck] (3.9,-6.1)-- (16.1,-6.1);
\draw [line width=2.pt,color=ffqqqq] (10.,-2.)-- (8.,-2.);
\draw [line width=2.pt,color=ffqqqq] (8.,-2.)-- (8.,0.);
\draw [line width=2.pt,color=ffqqqq] (8.,0.)-- (6.,0.);
\draw [line width=2.pt,color=ffqqqq] (6.,0.)-- (6.,2.);
\draw [line width=2.pt,color=ffqqqq] (6.,2.)-- (4.,2.);
\draw [line width=2.pt,color=ffqqqq] (4.,2.)-- (4.,-2.);
\draw [line width=1.pt,dotted,color=ffqqqq] (6.,0.)-- (4.,0.);
\draw [line width=1.pt,dotted,color=ffqqqq] (4.,-2.)-- (8.,-2.);
\draw [line width=1.pt,dotted,color=ffqqqq] (6.,0.)-- (6.,-4.);
\draw [line width=1.pt,dotted,color=ffqqqq] (8.,-2.)-- (8.,-6.);
\draw [line width=1.pt,dotted,color=ffqqqq] (6.,-4.)-- (10.,-4.);
\draw (2.884363351309636,3.08967050360574424) node[anchor=north west] {\LARGE $\widehat{X}_1$};
\draw (15.72160482015922,3.08967050360574424) node[anchor=north west] {\LARGE $\widehat{X}_2'$};
\draw [line width=2.pt,color=ffqqqq] (10.,-6.)-- (12.,-6.);
\draw [line width=2.pt,color=ffqqqq] (12.,-6.)-- (12.,0.);
\draw [line width=2.pt,color=ffqqqq] (12.,0.)-- (14.,0.);
\draw [line width=2.pt,color=ffqqqq] (14.,0.)-- (14.,-2.);
\draw [line width=2.pt,color=ffqqqq] (14.,-2.)-- (16.,-2.);
\draw [line width=2.pt,color=ffqqqq] (16.,-2.)-- (16.,2.);
\draw [line width=2.pt,color=ffqqqq] (16.,2.)-- (10.,2.);
\draw [line width=2.pt,color=ffqqqq] (10.,2.)-- (8.,2.);
\draw [line width=2.pt,color=ffqqqq] (8.,2.)-- (8.,0.);
\draw [line width=2.pt,color=ffqqqq] (8.,0.)-- (10.,0.);
\draw [line width=2.pt,color=ffqqqq] (10.,0.)-- (10.,-2.);
\draw [line width=1.pt,dotted,color=ffqqqq] (10.,0.)-- (10.,2.);
\draw [line width=1.pt,dotted,color=ffqqqq] (12.,0.)-- (12.,2.);
\draw [line width=1.pt,dotted,color=ffqqqq] (14.,0.)-- (14.,2.);
\draw [line width=1.pt,dotted,color=ffqqqq] (14.,0.)-- (16.,0.);
\draw [line width=1.pt,dotted,color=ffqqqq] (12.,0.)-- (10.,0.);
\draw [line width=1.pt,dotted,color=ffqqqq] (10.,-2.)-- (12.,-2.);
\draw [line width=1.pt,dotted,color=ffqqqq] (10.,-4.)-- (12.,-4.);
\begin{scriptsize}
\draw[color=ffqqqq] (5,1) node {\Large$\bullet$};
\draw[color=ffqqqq] (5,-1) node {\Large$\bullet$};
\draw[color=ffqqqq] (7,-1) node {\Large$\bullet$};
\draw[color=ffqqqq] (5,-3) node {\Large$\bullet$};
\draw[color=ffqqqq] (7,-3) node {\Large$\bullet$};
\draw[color=ffqqqq] (9,-3) node {\Large$\bullet$};
\draw[color=ffqqqq] (7,-5) node {\Large$\bullet$};
\draw[color=ffqqqq] (9,-5) node {\Large$\bullet$};
\draw[color=ffqqqq] (9,1) node {\Large$\bullet$};
\draw[color=ffqqqq] (11,1) node {\Large$\bullet$};
\draw[color=ffqqqq] (13,1) node {\Large$\bullet$};
\draw[color=ffqqqq] (15,1) node {\Large$\bullet$};
\draw[color=ffqqqq] (11,-1) node {\Large$\bullet$};
\draw[color=ffqqqq] (11,-3) node {\Large$\bullet$};
\draw[color=ffqqqq] (11,-5) node {\Large$\bullet$};
\draw[color=ffqqqq] (15,-1) node {\Large$\bullet$};
\draw [fill=ududff] (1.,3.) circle (2.5pt);
\draw[color=ududff] (1.171971740416364,3.4673321863738567) node {$A$};
\draw [fill=ududff] (1.,1.) circle (2.5pt);
\draw[color=ududff] (1.171971740416364,1.4577263081249454) node {$B$};
\draw [fill=ududff] (1.,-1.) circle (2.5pt);
\draw[color=ududff] (1.171971740416364,-0.5518795701239658) node {$C$};
\draw [fill=ududff] (1.,-3.) circle (2.5pt);
\draw[color=ududff] (1.171971740416364,-2.5366754992586933) node {$D$};
\draw [fill=ududff] (1.,-5.) circle (2.5pt);
\draw[color=ududff] (1.171971740416364,-4.546281377507604) node {$E$};
\draw [fill=ududff] (1.,-7.) circle (2.5pt);
\draw[color=ududff] (1.171971740416364,-6.531077306642333) node {$F$};
\draw [fill=ududff] (1.,5.) circle (2.5pt);
\draw[color=ududff] (1.171971740416364,4.9063092349965345) node {$G$};
\draw [fill=ududff] (21.,5.) circle (2.5pt);
\draw[color=ududff] (21.168790726448727,4.9063092349965345) node {$H$};
\draw [fill=ududff] (21.,3.) circle (2.5pt);
\draw[color=ududff] (21.168790726448727,3.4673321863738567) node {$I$};
\draw [fill=ududff] (21.,1.) circle (2.5pt);
\draw[color=ududff] (21.168790726448727,1.4577263081249454) node {$J$};
\draw [fill=ududff] (21.,-1.) circle (2.5pt);
\draw[color=ududff] (21.168790726448727,-0.5518795701239658) node {$K$};
\draw [fill=ududff] (21.,-3.) circle (2.5pt);
\draw[color=ududff] (21.168790726448727,-2.5366754992586933) node {$L$};
\draw [fill=ududff] (21.,-5.) circle (2.5pt);
\draw[color=ududff] (21.168790726448727,-4.546281377507604) node {$M$};
\draw [fill=ududff] (21.,-7.) circle (2.5pt);
\draw[color=ududff] (21.168790726448727,-6.531077306642333) node {$N$};
\draw [fill=ududff] (3.,7.) circle (2.5pt);
\draw[color=ududff] (-0.9120639851750981,5.204028624366743) node {$O$};
\draw [fill=ududff] (5.,7.) circle (2.5pt);
\draw[color=ududff] (-0.9120639851750981,5.204028624366743) node {$P$};
\draw [fill=ududff] (7.,7.) circle (2.5pt);
\draw[color=ududff] (-0.9120639851750981,5.204028624366743) node {$Q$};
\draw [fill=ududff] (9.,7.) circle (2.5pt);
\draw[color=ududff] (-0.9120639851750981,5.204028624366743) node {$R$};
\draw [fill=ududff] (11.,7.) circle (2.5pt);
\draw[color=ududff] (-0.9120639851750981,5.204028624366743) node {$S$};
\draw [fill=ududff] (13.,7.) circle (2.5pt);
\draw[color=ududff] (-0.9120639851750981,5.204028624366743) node {$T$};
\draw [fill=ududff] (15.,7.) circle (2.5pt);
\draw[color=ududff] (-0.9120639851750981,5.204028624366743) node {$U$};
\draw [fill=ududff] (17.,7.) circle (2.5pt);
\draw[color=ududff] (-0.9120639851750981,5.204028624366743) node {$V$};
\draw [fill=ududff] (19.,7.) circle (2.5pt);
\draw[color=ududff] (-0.9120639851750981,5.204028624366743) node {$W$};
\draw [fill=ududff] (3.,-9.) circle (2.5pt);
\draw[color=ududff] (3.1815776186652744,-8.540683184891243) node {$Z$};
\draw [fill=ududff] (5.023766397005074,-9.079754324721877) circle (2.5pt);
\draw[color=ududff] (5.253208369699644,-8.553088159448334) node {$A_1$};
\draw [fill=ududff] (7.,-9.) circle (2.5pt);
\draw[color=ududff] (7.238004298834369,-8.478658312105782) node {$B_1$};
\draw [fill=ududff] (9.,-9.) circle (2.5pt);
\draw[color=ududff] (9.24761017708328,-8.478658312105782) node {$C_1$};
\draw [fill=ududff] (11.,-9.) circle (2.5pt);
\draw[color=ududff] (11.232406106218006,-8.478658312105782) node {$D_1$};
\draw [fill=ududff] (13.,-9.) circle (2.5pt);
\draw[color=ududff] (13.242011984466915,-8.478658312105782) node {$E_1$};
\draw [fill=ududff] (15.,-9.) circle (2.5pt);
\draw[color=ududff] (15.226807913601641,-8.478658312105782) node {$F_1$};
\draw [fill=ududff] (17.,-9.) circle (2.5pt);
\draw[color=ududff] (17.23641379185055,-8.478658312105782) node {$G_1$};
\draw [fill=ududff] (19.,-9.) circle (2.5pt);
\draw[color=ududff] (19.24601967009946,-8.478658312105782) node {$H_1$};
\draw[color=black] (-0.7880142396041777,-6.580697204870701) node {$f$};
\draw[color=black] (-0.7880142396041777,-4.571091326621788) node {$g$};
\draw[color=black] (-0.7880142396041777,-3.1817341762274793) node {$h$};
\draw[color=black] (-0.7880142396041777,-1.172128297978568) node {$i$};
\draw[color=black] (-0.7880142396041777,0.8374775802703431) node {$j$};
\draw[color=black] (-0.7880142396041777,2.8222735094050706) node {$k$};
\draw[color=black] (-0.7880142396041777,4.807069438539798) node {$l$};
\draw[color=black] (3.2560074660078264,4.831879387653982) node {$m$};
\draw[color=black] (5.265613344256736,4.831879387653982) node {$n$};
\draw[color=black] (7.250409273391461,4.831879387653982) node {$p$};
\draw[color=black] (9.26001515164037,4.831879387653982) node {$q$};
\draw[color=black] (11.269621029889281,4.831879387653982) node {$r$};
\draw[color=black] (13.254416959024008,4.831879387653982) node {$s$};
\draw[color=black] (15.264022837272917,4.831879387653982) node {$t$};
\draw[color=black] (16.653379987667225,4.831879387653982) node {$a$};
\draw[color=black] (18.662985865916138,4.831879387653982) node {$b$};
\end{scriptsize}
\end{tikzpicture}
}
\caption{The corresponding connected set $\widehat{X}'=\widehat{X}_1\cup\widehat{X}_2'$, where $\widehat{X}_2'$ is obtained via translating the set $\widehat{X}_2$ by one unit to the left. The blue rectangle indicates the edge-boundary of $R(\widehat{X}')$.}
\label{fig:z2connected}
\end{figure}
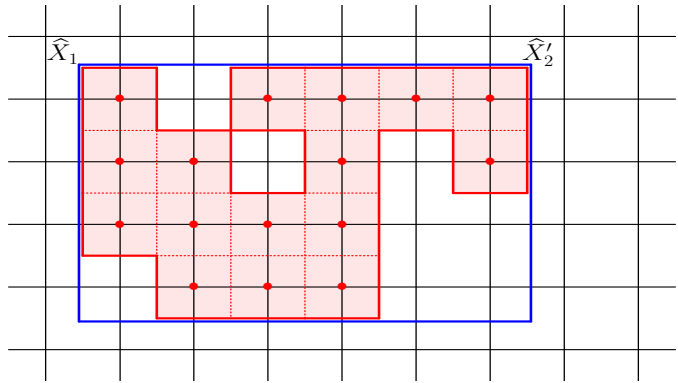}

Now, let $R(\widehat{X})$ denote the smallest rectangular\footnote{With a subset $R(\widehat{X})$ of $\Z^2$ being rectangular, we mean that there exist intervals $I_1,I_2\subset\R$ such that $R(\widehat{X})=(I_1\cap\Z)\times(I_2\cap\Z).$} subset of $\Z^2
$ that contains $\widehat{X}$ and observe that 
\begin{equation} \label{eq:surfestim}
S(\widehat{X})\geq S(R(\widehat{X})) \quad \mbox{(i.e., Per$(\widehat{X})\geq\;$Per$(R(\widehat{X}))$)}
\end{equation}
as well as 
\begin{equation}\label{eq:volestim}
|\widehat{X}|\leq |R(\widehat{X})| \quad \mbox{(i.e., Area$(\widehat{X})\leq\;$Area$(R(\widehat{X}))$)}\:,
\end{equation}
for which we visualized a special case in Figure \ref{fig:z2connected}. Note that for \eqref{eq:surfestim}, we are using that $\widehat{X}$ is connected. There are now two possibilities: either $\mathcal{U}(R(\widehat{X}))$ is a square or it is not. If $\mathcal{U}(R(\widehat{X}))$ is a square, its area must be bigger than $N$ since $\widehat{X}$ was assumed to be a non-quadratic droplet configuration and is thus a proper subset of $R(\widehat{X})$ with $|\widehat{X}|=N$. This immediately implies that $\mbox{Per}(R(\widehat{X}))>\mbox{Per}(Q_x)=4\sqrt{N}$ and thus by \eqref{eq:surfestim}, we get that 
\begin{equation} 
S(\widehat{X})\geq S(R(\widehat{X}))>S(Q_x)=4\sqrt{N}\:,
\end{equation}
which is a contradiction since $\widehat{X}$ was assumed to be an element of $\mathcal{V}_{N,1}$. 

If $\mathcal{U}(R(\widehat{X}))$ is a non-square rectangle with $\mbox{Area}(R(\widehat{X}))\geq \mbox{Area}(\widehat{X})=N$, then an easy calculus exercise to shows that 
\begin{equation} 
\mbox{Per}(R(\widehat{X}))>\mbox{Per}(Q_x)\:.
\end{equation}
By \eqref{eq:surfestim} we again get the contradiction 
\begin{equation} 
S(\widehat{X})\geq S(R(\widehat{X}))>S(Q_x)=4\sqrt{N}\:.
\end{equation}
This shows the lemma.
\end{proof}

\subsection{Classical droplets on the strip} \label{sec:stripdrop}

Let $\mathcal{G}$ be now the strip of width $M$ as described in Section \ref{sec:strip}. For sufficiently large particle number that are multiples of $M$, let us now show that the classical droplets are given by rectangular configurations $\mathcal{R}_{z,\ell}$ (cf.\ \eqref{eq:gluehwein}).

\begin{lemma} \label{lemma:strip}
Let $\mathcal{G}$ be the strip of width $M$. Let  $N$ is a multiple of $M$, i.e.\ $N=\ell M$, for some $\ell\in\N$ and assume that $\ell>\frac{M}{2}$. Then $S_{N,min}=2M$ and 
\begin{equation}
\mathcal{V}_{N,1}=\{\mathcal{R}_{z,\ell}: z \in \Z\}\:.
\end{equation}
\end{lemma}

The reason for including the assumption $\ell> \frac{M}{2}$ is that for smaller $\ell$ the minimizers will be rectangles attached to the floor or ceiling, but not both. For particle numbers $N$ which are not multiples of $M$ it is quite convincing that minimizing configurations will be found by attaching surplus particles in an extra column at one of the sides of a region of the form $R_{z,\ell}$, starting from either the ceiling or the floor of the strip. But we haven't worked out a proof for this and don't know a reference.

\begin{proof}
Firstly, note that for any $z\in\Z$, we have that 
\begin{equation} 
|\mathcal{R}_{z,\ell}|=kM=N\quad\text{and}\quad S\left(\mathcal{R}_{z,\ell}\right)=2M\:.
\end{equation}
Now, for a proof by contradiction, assume that there exists a configuration $\widehat{X}\in\mathcal{V}_N$ such that $S(\widehat{X})\leq 2M$ and which is not of the form $\widehat{X}= \mathcal{R}_{z,\ell}$ for some $z\in\Z$. 

Now, note that any configuration $\widehat{X}\in\mathcal{V}_{N,1}$ will be a connected set that is attached to the ``floor" or the ``ceiling" of the strip, i.e.\ there exists $x=(z,m)\in\widehat{X}$ such that $m=1$ or $m=M$ respectively. For if $\widehat{X}$ is not attached to either floor or ceiling, it could be translated in $m$-direction until it is, which would reduce its edge surface by a strictly positive number.

Likewise, it is impossible that $\widehat{X}$ is not connected, which follows from an argument similar as in the proof of Lemma \ref{lemma:square}. The only situation that requires extra care is the case that $\widehat{X}$ is not connected and could be decomposed into two disjoint non-connected sets $\widehat{X}=\widehat{X}_1\cup\widehat{X}_2$ such that $\widehat{X}_1$ is attached to the floor while $\widehat{X}_2$ is attached to the ceiling. In this case, however, we can reflect $\widehat{X}_2$ about the $m=\frac{M+1}{2}$ axis to obtain a configuration $\widehat{X}_2'$ with same edge-surface as $\widehat{X}_2$ (cf.\ Figure~\ref{fig:B3}) which is attached to the floor as well. By translation of $\widehat{X}_2'$ in $z$-direction, we obtain a set $\widehat{X}_2''$ such that $\widehat{X}_1$ and $\widehat{X}_2''$ are disjoint but connected and thus $\widehat{X}''=\widehat{X}_1\cup\widehat{X}_2''$ has a strictly smaller edge boundary than $\widehat{X}=\widehat{X}_1\cup\widehat{X}_2$ (cf.\ Figure~\ref{fig:B3}).

\definecolor{ffqqqq}{rgb}{1.,0.,0.}
\definecolor{qqqqff}{rgb}{0.,0.,1.}
\definecolor{qqffqq}{rgb}{0.,1.,0.}
\definecolor{zzttqq}{rgb}{0.6,0.2,0.}

\begin{figure}[h]
\centering
\resizebox{9cm}{4cm}
{
\begin{tikzpicture}[line cap=round,line join=round,>=triangle 45,x=1.0cm,y=1.0cm, scale=1]
\clip(-1.,2.898968474833424) rectangle (19.004906413331508,11.089205234803678);
\fill[line width=2.pt,color=zzttqq,fill=zzttqq,fill opacity=0.10000000149011612] (2.,3.) -- (2.,6.) -- (1.,6.) -- (1.,7.) -- (6.,7.) -- (6.,5.) -- (4.,5.) -- (4.,3.) -- cycle;
\fill[line width=2.pt,color=qqffqq,fill=qqffqq,fill opacity=0.10000000149011612] (17.000000033211688,10.99974227266931) -- (17.,9.) -- (15.,9.) -- (15.,8.) -- (13.,8.) -- (13.,10.) -- (14.,10.) -- (14.000000083029223,10.999355681673276) -- cycle;
\fill[line width=2.pt,color=qqqqff,fill=qqqqff,fill opacity=0.10000000149011612] (17.,3.) -- (17.,5.) -- (15.,5.) -- (15.,6.) -- (13.,6.) -- (13.,4.) -- (14.,4.) -- (14.,3.) -- cycle;
\fill[line width=2.pt,color=ffqqqq,fill=ffqqqq,fill opacity=0.10000000149011612] (6.,6.) -- (8.,6.) -- (8.,5.) -- (10.,5.) -- (10.,3.) -- (7.,3.) -- (7.,4.) -- (6.,4.) -- cycle;
\draw [line width=5.2pt,domain=-1.:19.004906413331508] plot(\x,{(--39.-0.*\x)/13.});
\draw [line width=5.2pt,domain=-1.:19.004906413331508] plot(\x,{(--196.35003668836413--0.0023007289968681732*\x)/17.853977321754964});
\draw [line width=2.pt,color=zzttqq] (2.,3.)-- (2.,6.);
\draw [line width=2.pt,color=zzttqq] (2.,6.)-- (1.,6.);
\draw [line width=2.pt,color=zzttqq] (1.,6.)-- (1.,7.);
\draw [line width=2.pt,color=zzttqq] (1.,7.)-- (6.,7.);
\draw [line width=2.pt,color=zzttqq] (6.,7.)-- (6.,5.);
\draw [line width=2.pt,color=zzttqq] (6.,5.)-- (4.,5.);
\draw [line width=2.pt,color=zzttqq] (4.,5.)-- (4.,3.);
\draw [line width=2.pt,color=qqffqq] (17.000000033211688,10.99974227266931)-- (17.,9.);
\draw [line width=2.pt,color=qqffqq] (17.,9.)-- (15.,9.);
\draw [line width=2.pt,color=qqffqq] (15.,9.)-- (15.,8.);
\draw [line width=2.pt,color=qqffqq] (15.,8.)-- (13.,8.);
\draw [line width=2.pt,color=qqffqq] (13.,8.)-- (13.,10.);
\draw [line width=2.pt,color=qqffqq] (13.,10.)-- (14.,10.);
\draw [line width=2.pt,color=qqffqq] (14.,10.)-- (14.000000083029223,10.999355681673276);
\draw [line width=2.pt,color=qqqqff] (17.,3.)-- (17.,5.);
\draw [line width=2.pt,color=qqqqff] (17.,5.)-- (15.,5.);
\draw [line width=2.pt,color=qqqqff] (15.,5.)-- (15.,6.);
\draw [line width=2.pt,color=qqqqff] (15.,6.)-- (13.,6.);
\draw [line width=2.pt,color=qqqqff] (13.,6.)-- (13.,4.);
\draw [line width=2.pt,color=qqqqff] (13.,4.)-- (14.,4.);
\draw [line width=2.pt,color=qqqqff] (14.,4.)-- (14.,3.);
\draw [line width=2.pt,color=ffqqqq] (6.,6.)-- (8.,6.);
\draw [line width=2.pt,color=ffqqqq] (8.,6.)-- (8.,5.);
\draw [line width=2.pt,color=ffqqqq] (8.,5.)-- (10.,5.);
\draw [line width=2.pt,color=ffqqqq] (10.,5.)-- (10.,3.);
\draw [line width=2.pt,color=ffqqqq] (7.,3.)-- (7.,4.);
\draw [line width=2.pt,color=ffqqqq] (7.,4.)-- (6.,4.);
\draw [line width=2.pt,color=ffqqqq] (6.,4.)-- (6.,6.);
\draw [line width=2.pt] (0.95,3.)-- (0.95,7.05);
\draw [line width=2.pt] (0.95,7.05)-- (10.05,7.05);
\draw [line width=2.pt] (10.05,7.05)-- (10.05,3.);
\draw (14.873548047682794,10.12884561471867) node[anchor=north west] {\Large $\widehat{X}_2$};
\draw (14.49302819821515,4.82087571639292) node[anchor=north west] {\Large $\widehat{X}_2'$};
\draw (7.933590793106227,4.394888001731736) node[anchor=north west] {\Large $\widehat{X}_2''$};
\draw (3.548552527812416,6.160567184556092) node[anchor=north west] {\Large $\widehat{X}_1$};
\draw (8.332230635405663,7.964526589594877) node[anchor=north west] {\Large ${R}(\widehat{X}'')$};
\begin{scriptsize}
\draw[color=black] (-3.373284733932359,3.306668313548758) node {$f$};
\draw[color=black] (-3.373284733932359,10.86270532440627) node {$g$};
\end{scriptsize}
\end{tikzpicture}
}
\caption{Disconnected $\widehat{X}=\widehat{X}_1\cup\widehat{X}_2$ with $\widehat{X}_1$) attached to the floor and $\widehat{X}_2$ attached to the ceiling.  ${R}(\widehat{X}'')$ is the smallest rectangle containing $\widehat{X}''=\widehat{X}_1\cup\widehat{X}_2''$.} 
\label{fig:B3}
\end{figure}
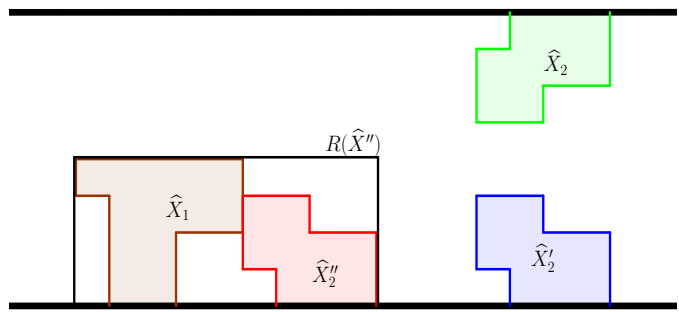

Hence, without loss of generality, assume from now on that $\widehat{X}$ is a connected configuration that is attached to the floor. As in the proof of Lemma \ref{lemma:square}, let ${R}(\widehat{X})\subset\mathcal{V}$ be the smallest rectangular configuration containing $\widehat{X}$ (compare Figure~\ref{fig:B3}) and note that since $\widehat{X}$ is connected, we have $S(\widehat{X})\geq S(R(\widehat{X}))$
and, trivially, $N=|\widehat{X}|\leq |R(\widehat{X})|$. While $R(\widehat{X})$ itself always has to be attached to the floor, it is now either possible that it is attached to the ceiling or not -- depending on whether $\widehat{X}$ is. Let us firstly treat the case that $\widehat{X}$ and therefore $R(\widehat{X})$ are not attached to the ceiling. In this case, using that $N=\ell M>\frac{M^2}{2}$, it is a simple calculus exercise to show that $S(R(\widehat{X}))=
\mbox{Per}(R(\widehat{X}))>2M$ and thus we get $S(\widehat{X})\geq S(R(\widehat{X}))>2M$ -- a contradiction to $\widehat{X}\in\mathcal{V}_{N,1}$. 

The case that remains to be treated is when the $\widehat{X}$ is attached to both -- floor and ceiling. Since $\widehat{X}$ was assumed to be non-rectangular, it is clear that $\partial(\mathcal{U}(\widehat{X}))$ cannot just be given by two straight vertical lines that connect floor and ceiling, which means that $\mbox{Per}(\widehat{X})=S(\widehat{X})$ has to be strictly larger than $2M$. This finishes the proof. 	

\end{proof}

\end{appendix}

\end{document}